\documentclass[journal]{IEEEtran}
\usepackage{cite}

\usepackage{graphicx}
\usepackage{subfigure}
\usepackage{epstopdf}
\usepackage{enumerate}
\usepackage{mathrsfs}
\usepackage{bbm}
\usepackage{algorithm}
\usepackage{algorithmic}
\usepackage{amsthm,amsmath,amsfonts}
\newtheorem{lemma}{\textbf{Lemma}}
\usepackage{bm}
\usepackage{subfigure}

\newtheorem{mydef}{Definition}

\usepackage{color}

\begin{document}
%
\title{RAN Slicing for Massive IoT and Bursty URLLC Service Multiplexing: Analysis and Optimization}
%
%
%

\author{Peng Yang,~\IEEEmembership{Member,~IEEE}, Xing Xi,~\IEEEmembership{Graduate Student Member,~IEEE}, Tony Q. S. Quek,~\IEEEmembership{Fellow,~IEEE}, Jingxuan Chen, Xianbin Cao,~\IEEEmembership{Senior Member,~IEEE},
        Dapeng Wu,~\IEEEmembership{Fellow,~IEEE}
\thanks{
P. Yang and T. Q. S. Quek are with the Information Systems Technology and Design, Singapore University of Technology and Design, 487372 Singapore.

X. Xi, J. Chen, and X. Cao are with the School of Electronic and Information Engineering, Beihang University, Beijing 100083, China, and also with the Key Laboratory of Advanced Technology, Near Space Information System (Beihang University), Ministry of Industry and Information Technology of China, Beijing 100083, China.

D. Wu is with the Department of Electrical and Computer Engineering, University of Florida, Gainesville FL 32611 USA.}
}

\maketitle

\begin{abstract}
Future wireless networks are envisioned to serve massive Internet of things (mIoT) via some radio access technologies,
where the random access channel (RACH) procedure should be exploited for IoT devices to access the networks.
However, the theoretical analysis of the RACH procedure for massive IoT devices is challenging.
To address this challenge, we first correlate the RACH request of an IoT device with the status of its maintained queue and analyze the evolution of the queue status. Based on the analysis result, we {then} derive the closed-form expression of the random access (RA) success probability, {which is a significant indicator characterizing the RACH procedure} of the device. Besides, considering the agreement on converging different services onto a shared infrastructure, we investigate the RAN slicing for mIoT and bursty ultra-reliable and low latency communications (URLLC) service multiplexing. Specifically, we formulate the RAN slicing problem as an optimization one to maximize the total RA success probabilities of all IoT devices and provide URLLC services for URLLC devices in an energy-efficient way. A slice resource optimization (SRO) algorithm exploiting relaxation and approximation with provable tightness and error bound is then proposed to mitigate the optimization problem. {Simulation results demonstrate that the proposed SRO algorithm can effectively implement the service multiplexing of mIoT and bursty URLLC traffic.}
\end{abstract}


%
\IEEEpeerreviewmaketitle

\section{Introduction}
%
%
%
%
\IEEEPARstart{W}{ith} the explosive growth of the Internet of Things (IoT), massive IoT (mIoT) devices, the number of which is predicted to reach 20.8 billion by 2020, will access the wireless networks for implementing advanced applications{\cite{says20156}}.
These applications include e-health, public safety, smart traffic, virtual navigation/management, and environment monitoring{\cite{liu2018novel}}.
To address the IoT market, the third-generation partnership project (3GPP) has identified mIoT as one of the three main use cases of 5G and has already initiated several task groups to standardize several solutions including extended coverage GSM (EC-GSM), LTE for machine-type communication (LTE-M), and narrowband IoT (NB-IoT) \cite{Cellular-Ericsson,LTE-Nokia}.
For establishing massive connections among wireless networks and IoT devices, the research on reliable and efficient access mechanisms should be prioritized.

In accomplishing the massive connections, when an IoT device wants to transmit signals in the uplink, it randomly chooses a random access (RA) preamble from an RA preamble pool and transmits it through an RA channel (RACH). If more than one device tries to access a base station (BS) simultaneously, then interference occurs at the BS. {Unless interference is resolved, the grand goal of accomplishing the massive connections cannot be achieved.
}

\subsection{Prior arts}
During the past few years, a rich body of works \cite{xing2019novel,jang2016non,zhang2019dnn,moon2018sara,miuccio2020joint,di2018learning,vilgelm2019dynamic,seo2019low} on RA mechanisms has been developed to mitigate interference and improve the RA success probability or reduce the access delay of an IoT device. For example, the work in \cite{xing2019novel} proposed to improve the RA success probability of an IoT device by exploiting a distributed queue mechanism. It {also designed} an access resource grouping mechanism to reduce the access delay caused by the queuing process of the distributed queue mechanism.
The work in \cite{jang2016non} proposed a novel scheme to increase RA success probability. First, this scheme increased the number of preambles at the first step of the RACH procedure by utilizing a spatial group mechanism. Second, it improved resource utilization through non-orthogonally allocating uplink channel resources at the second step of the RACH procedure.
Besides, to reduce the access delay, a grant-free non-orthogonal RA system relying on the accurate user activity detection and channel estimation was proposed in \cite{zhang2019dnn}.

Most of the studies \cite{xing2019novel,jang2016non,zhang2019dnn,moon2018sara,miuccio2020joint,di2018learning,vilgelm2019dynamic,seo2019low}, however, assumed that network resources were reserved for the IoT service and did not study the case of the coexistence of IoT service and many other services such as ultra-reliable and low latency communications (URLLC). The research of the coexistence of IoT service and other services is essential. {This is because} future networks are convinced to integrate {heterogeneous} services with different latency, reliability, and throughput requirements into a shared physical infrastructure rather than deploying individual network solution for each type of service \cite{alliance20155g}.
What is more, owing to the shared characteristic of network resources, some conclusions obtained {from IoT dedicated networks} may become inapplicable if {heterogeneous} services are required to be supported by the networks.

Network slicing is considered as a promising technique in future networks to converge {heterogeneous} services onto a shared physical infrastructure. This can be implemented by {logically} partitioning the infrastructure into multiple network slices, where a network slice {capable of providing a negotiated service quality} is defined as an end-to-end virtual network running on the infrastructure \cite{rost2017network}.
Recently, many slicing frameworks have been developed to provide performance guarantees for IoT or massive machine-type communications (mMTC) service, enhanced mobile broadband (eMBB) service, and URLLC service{\cite{popovski20185g,budhiraja2019tactile,ksentini2017toward,How2019yang,d2019slice,zambianco2020interference}}.
For example, instead of slicing the radio access network (RAN) via orthogonal resource allocation, the work in \cite{popovski20185g,budhiraja2019tactile} studied the advantages of allowing for non-orthogonal RAN resources sharing {among} a set of mMTC, eMBB, and URLLC users.
The work in \cite{ksentini2017toward} developed a two-level scheduling process to allocate dynamically dedicated bandwidth to each slice based on workload demand and slices' quality of service (QoS) requirements. In \cite{How2019yang}, we proposed to orchestrate network resources for a network slicing system to guarantee more reliable URLLC and energy-efficient eMBB service provisions. {Besides, the work in \cite{d2019slice,zambianco2020interference} proposed to mitigate the inter-slice interference via RAN slicing such that heterogeneous services could be supported by the same physical infrastructure. However, the significant impact of intra-cell interference on the obtained results was not considered.}

\subsection{Motivation and contributions}
Unlike the work in \cite{xing2019novel,jang2016non,zhang2019dnn,moon2018sara,miuccio2020joint,di2018learning,vilgelm2019dynamic,seo2019low,popovski20185g,
budhiraja2019tactile,ksentini2017toward,How2019yang,d2019slice,zambianco2020interference},
this paper simultaneously analyzes the {RACH procedure} for the mIoT service and studies the RAN slicing for the mIoT service included service multiplexing.
This study is highly challenging because i) performance requirements of a massive number of IoT devices should be satisfied. Yet, the typical 5G cellular IoT, NB-IoT can only admit 50,000 devices per cell \cite{3GPP15Cellular}, and the 5G new radio (NR) technique can connect a great number of devices {only} by deploying costly ultra-dense heterogeneous networks;
ii) RAN slicing operation (e.g., activating and releasing slices) has to be conducted in a timescale of minutes to hours to keep pace with the upper layer network slicing. In the process of slicing upper layer networks, some functions (e.g., radio resource control function) and protocols (e.g., RAN protocol stacks) should be activated and configured, which are time-consuming \cite{rost2017network}. However, wireless channels generally change in a timescale of a millisecond to seconds, {which is much shorter than the RAN slicing operation duration. As a result, how to optimally perform the time-consuming RAN slicing operation in rapidly changing wireless channels, which is called a two-timescale issue of the RAN slicing \cite{tang2019service}, is a big challenge;}
iii) compared with the resource allocation problem for other services (e.g., eMBB service), the resource allocation problem for the bursty URLLC service {(e.g., human-machine collaborations in industry automation, telesurgery in healthcare, virtual reality for remote education or in gaming industry \cite{hou2019prediction,antonakoglou2018toward,ruan2019machine})} where URLLC packets arrive in a burst may be more challenging due to the stringent low latency requirement and the 99.999\% reliability requirement \cite{series2015imt}; iv)
{IoT devices may experience radio shadowing and then the transmission outage in challenging radio environments.}

These challenges motivate us to investigate the RAN slicing for mIoT and bursty URLLC service multiplexing to maximize
{the total RA success probabilities of all IoT devices while providing URLLC services for URLLC devices.}
{Besides, to overcome the radio shadowing in challenging radio environments, the coordinated multi-point (CoMP) transmission technique, which creates spatial diversity with redundant communication paths, is exploited. CoMP can also significantly improve transmission reliability via spatial diversity instead of relying on {the} packet retransmission.}
{Therefore, this paper considers the RAN incorporating the CoMP transmission technique, which is called CoMP-enabled RAN.}
Summarily, the main contributions of this paper are presented as the following:
\begin{itemize}
\item {The subframe structure designed for NB-IoT is extended} for mIoT transmissions to accommodate more RACH requests from a massive number of IoT devices;
\item {We analyze the queue evolution processes including the IoT packet arrival, accumulation and departure processes} by employing probability and stochastic geometry theories. Based on the analysis result, we derive the closed-form expression of the RA success probability of a randomly chosen IoT device;
\item We define the mIoT slice utility {as the time average of RA success probability of all IoT devices} and the bursty URLLC slice utility {as the time-average energy efficiency, which reflects the transmission latency and power consumption, for serving URLLC devices}. {We then} formulate the CoMP-enabled RAN slicing for mIoT and bursty URLLC service multiplexing as a resource optimization problem. The objective of the optimization problem is to maximize the total mIoT and URLLC slice utilities, subject to limited physical resource constraints. The solution of this problem is difficult due to the existence of {the} indeterministic objective function and thorny non-convex constraints and the requirement of tackling a two-timescale issue as well;
\item To mitigate this thorny optimization problem, we propose a slice resource optimization (SRO) algorithm. In this algorithm, we first exploit a sample average approximate (SAA) technique and an alternating direction method of multipliers (ADMM) to tackle the {indeterminacy of the} objective function and the two-timescale issue. Then, a semidefinite relaxation (SDR) scheme joint with a Taylor expansion scheme is leveraged to approximate the non-convex problem as a convex one. The tightness of the SDR scheme and the error bound of the Taylor expansion are also analyzed.
\end{itemize}

\subsection{Organization}
{The remainder of this paper is organized as follows: Section \uppercase\expandafter{\romannumeral2} builds the system model. Based on the model, a CoMP-enabled RAN slicing problem for mIoT and bursty URLLC service multiplexing is formulated in Section \uppercase\expandafter{\romannumeral3}. Section \uppercase\expandafter{\romannumeral4} aims to derive the closed-form expression of the RA success probability. Section \uppercase\expandafter{\romannumeral5} and Section \uppercase\expandafter{\romannumeral6} propose to mitigate the formulated problem with system generated channel coefficients and sensed channel coefficients, respectively. The simulation is conducted in Section \uppercase\expandafter{\romannumeral7}, and Section \uppercase\expandafter{\romannumeral8} concludes this paper.}

{\emph{Notation:} Boldface uppercase letters denote matrices, whereas boldface lowercase letters denote vectors. The superscripts $(\cdot)^{\rm T}$ and $(\cdot)^{\rm H}$ denote transpose and conjugate transpose matrix operators. ${\rm tr}(\cdot)$, ${\rm rank}(\cdot)$, and $\left\lfloor {\cdot} \right\rfloor$ denote the trace, the rank, and the rounding down operators, respectively. ${\bm X} \succeq 0$ indicates that $\bm X$ is a Hermitian positive-semidefinite matrix.
For clarification, some significant notations are listed in Table \uppercase\expandafter{\romannumeral1}.}
\begin{table*}[!t]
\renewcommand{\arraystretch}{1.2}
\caption{{{List of Notations}}}
\label{table_1}
\newcommand{\tabincell}[2]{\begin{tabular}{@{}#1@{}}#2\end{tabular}}
\centering
\begin{tabular}{|c||c|c||c|}
\hline
{Notation} & {Description} & {Notation} & {Description} \\\hline
{{${\mathcal S}^I$, ${\mathcal S}^u$}} & {mIoT, URLLC slice sets} &	{${\bm u}_{i,s}$} & {Location of the $i$-th IoT device in $s\in {\mathcal S}^I$} \\ \cline{1-4}
{$\lambda_{s}^{I}$} &  {Intensity of IoT devices in $s \in {\mathcal S}^I$}  & {$N^u$} & {Number of URLLC devices} \\ \cline{1-4}
{$\lambda_R$} & {Intensity of RRHs}   &  {${\bm v}_j$} & {Location of the $j$-th RRH} \\ \cline{1-4}
{$K$}  & {Number of antennas of an RRH} & {$W$} & {Total system bandwidth} \\ \cline{1-4}
{$\theta_s^{\rm th}$}  & {SINR threshold for decoding an IoT packet in $s\in{\mathcal S}^I$} & {$I_s^u$} & {Number of URLLC devices in $s \in {\mathcal S}^u$} \\ \cline{1-4}
{$D_s$} & {Transmission latency requirement of a URLLC device}  & {$\alpha$} & {Blocking probability threshold of a URLLC packet} \\\hline
{$\beta$} & {Codeword error decoding probability threshold} & {$\vartheta_{w,s}(t)$} & {Intensity of new IoT arrival packets} \\ \cline{1-4}
{$N_{a,s}(t)$} & {Accumulated number of IoT packets} &  {$L$}  & {Number of information bits of an IoT packet} \\ \cline{1-4}
{$N_{w,s}(t)$} & {Number of new IoT arrival packets} &  {$a$}	 &  {Size of the tone spacing} \\ \hline
{$P_{ne,s}(t)$} & {Non-empty probability of a queue} &  {$F_s$}	 &  {Number of orthogonal uplink PRACHs} \\ \hline
{$\xi $} & {Number of non-dedicated RA preambles} &  {$P_{ACB}$}	 &  {An ACB factor} \\ \hline
{$P_{nr,s}(t) $} & {Probability without restricting RACH requests} &  {$P_{s}(t)$}	 &  {RA success probability of an IoT device} \\ \hline
{$\rho_o $} & {Power cutoff threshold} &  {${\mathcal I}_s(t)$}	 &  {Intra-cell interference} \\ \hline
{$L_{{\mathcal I}_s(t)}(\cdot) $} & {Laplace transform of the PDF of ${\mathcal I}_s(t)$} &  {$\vartheta_{a,s}^t$}	 &  {Intensity of accumulated IoT packets} \\ \hline
{$\lambda_s $} & {Intensity of new URLLC arrival packets} &  {$W^u(\bm r(t))$}	 &  {Bandwidth allocated to URLLC slices} \\ \hline
{$\omega_{i,s}^u(t)$} & {Bandwidth allocated for transmitting a URLLC packet} &  {$b_{i,s}^u(t)$}	 &  {Indicator of whether a URLLC device is served} \\ \hline
{$r_{i,s}^u(t)$} & {Channel uses for transmitting a URLLC packet} &  {$\varsigma $}	 &  {Queueing probability of a URLLC packet} \\ \hline
{$\bm g_{ij,s}(t)$} & {\tabincell{l}{Transmit beamformer pointing at the $i$-th \\ URLLC device from the $j$-th RRH}} &  {$\bm h_{ij,s}(t) $}	 &  {\tabincell{l}{Channel coefficient between the $i$-th \\ URLLC device and the $j$-th RRH}} \\ \hline
{$E_j$} & {Maximum transmit power of the $j$-th RRH} &  {$L_{i,s}^u(t) $}	 &  {Number of information bits of a URLLC packet} \\ \hline
{$\hat E_j^I$} & {Transmit power of a RRH for connecting to an IoT device} &  {$\omega_s(\bar t) $}	 &  {Bandwidth allocated to a mIoT slice $s\in {\mathcal S}^I$} \\ \hline
{$\bm H_{i,sm}$} & {Channel matrix corresponding to the $m$-th channel sample} &  {$\bm G_{i,sm}$} & {Power matrix corresponding to the $m$-th channel sample} \\ \hline
\end{tabular}
\end{table*}

\begin{figure}[!t]
\centering
\includegraphics[width=3.3in]{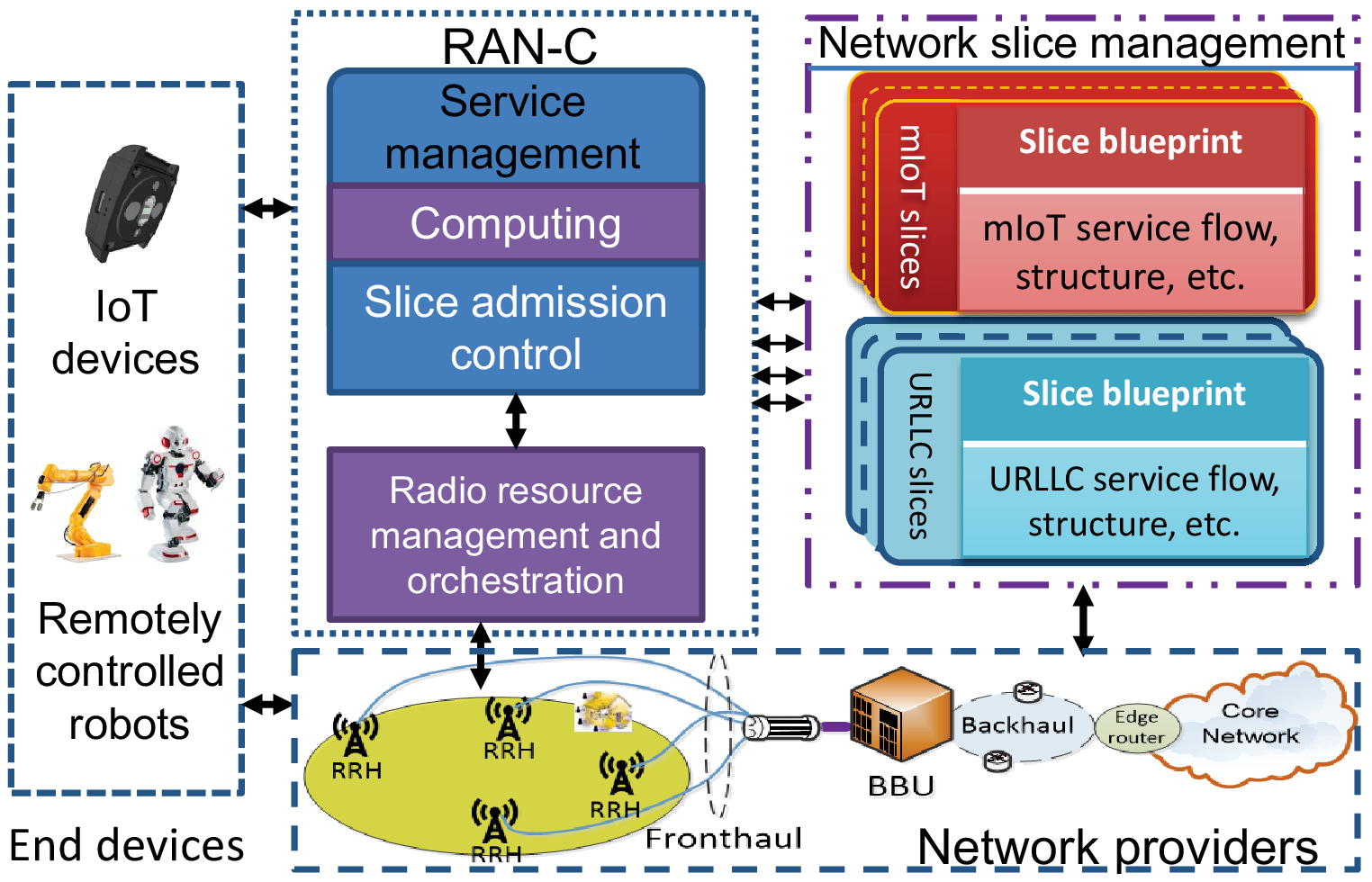}
\caption{{A CoMP-enabled RAN slicing system.}}
\label{fig:fig_URLLC_mIoT_slice}
\end{figure}

\section{System model}
{As shown in Fig. \ref{fig:fig_URLLC_mIoT_slice},} we consider a CoMP-enabled RAN slicing system for mIoT and bursty URLLC multiplexing service provision. From the viewpoint of infrastructure composition, the system mainly includes one baseband unit (BBU) and multiple remote radio heads (RRHs) connecting to the BBU via fronthaul links\footnote{{As in \cite{matera2018non}, we assume the perfect network synchronization and the available low-latency backhaul for the coordination. Although the CoMP structure requires the data sharing among RRHs resulting in additional overhead, some schemes such as short packet communication, flexible subcarrier spacing, and a flexible number of symbols in one transmission time interval \cite{khan2020availability} can be exploited to satisfy the stringent low latency requirement of URLLC.
}}.
From the perspective of network slicing, two types of inter-slices, i.e., mIoT slices and URLLC slices\footnote{{The service multiplexing of eMBB traffic is not considered in this paper as it has been investigated in our paper \cite{tang2019service}.}}, are exploited in this system with $\mathcal{S}^I$ and $\mathcal{S}^u$ denoting the mIoT slice set and {the} URLLC slice set. We focus on the modelling of uplink IoT data transmission in mIoT slices and the modelling of downlink URLLC data transmission in URLLC slices.
IoT devices (e.g., water meters and wearable e-health devices) are spatially distributed in ${\mathbb R}^2$ according to an independent homogeneous Poisson point process (PPP) $\Phi_s = \{{\bm u}_{i,s}; s \in \mathcal{S}^I, i = 1, 2, \ldots\}$ with intensity $\lambda_s^I$, where ${\bm u}_{i,s}$ is the $i$-th IoT device's location in the $s$-th mIoT slice.
There are $N^u$ URLLC devices (e.g., remote-controlled robot sensors)
that are randomly and evenly distributed in ${\mathbb R}^2$.
The RRHs are spatially distributed in ${\mathbb R}^2$ according to an independent PPP $\Phi_R = \{\bm v_j; j = 1, 2, \ldots\}$ with intensity $\lambda_R$, where $\bm v_j$ represents the location of the $j$-th RRH.
The number and locations of IoT devices and RRHs will be fixed once deployed. Besides, each RRH is equipped with $K$ antennas, and each device is equipped with a single antenna. The total system bandwidth $W$ of the system is limited and shared by mIoT slices and URLLC slices. A flexible frequency division multiple access (FDMA) technique is utilized to achieve the inter-slice and intra-slice interference isolation \cite{tang2019service}.

{\subsection{CoMP-enabled RAN slicing system architecture}}
In view of the architecture, the CoMP-enabled RAN slicing system consists of four parts including end devices, RAN coordinator (RAN-C), network slice management, and network providers, as shown in Fig. \ref{fig:fig_URLLC_mIoT_slice}. The system time is discretized and partitioned into time slots and minislots with a time slot consisting of $T$ minislots. At the beginning of each time slot, the RAN-C will decide whether to accept or reject received network slice requests defined later after checking available resource information (e.g., physical resource blocks (PRBs)) and computing. If a slice request is accepted, network slice management will be responsible for creating or activating corresponding types of virtual slices and configuring RAN protocol stacks, the processes of which are time-consuming and usually in a timescale of minutes to hours. Next, if a slice request admission arrives, network providers will find the optimal servers and paths to deploy virtual network functions to satisfy the end-to-end QoS requirements of the slice\footnote{{Although network providers are included as part of the architecture, the related problem of finding optimal servers and paths to deploy virtual network functions is not studied in this paper.}}. Meanwhile, at the beginning of each minislot, each active IoT device may try to connect to its associated RRH. RRHs will generate cooperated beamformers pointing at URLLC devices based on sensed channel coefficients.


{\subsection{mIoT and bursty URLLC slice model}}
{According to} the above mentioned network slice concept, especially from the viewpoint of the slice's QoS requirement, we can define a mIoT slice request as follows.

\begin{mydef}\label{def:IoT_slice_definition}
A mIoT slice request is defined as a tuple $\{\lambda_s^I, \theta_s^{th}\}$ for any slice $s \in {\mathcal S}^I$, where $\theta_s^{th}$ is the signal-to-interference-plus-noise ratio (SINR) threshold for an RRH to successfully decode packets (including preamble packets and IoT data packets) sent from an IoT device in $s$.
\end{mydef}

%
In this paper, RRHs will assign IoT devices to {$|{\mathcal S}^I|$} different slices according to the received SINR {with $|\cdot|$ indicating the number of elements in a set}. The SINR threshold configured for all IoT devices in a slice is similar.

\begin{mydef}\label{mydef:URLLC_slice_request}
A bursty URLLC slice request is defined as four tuples $\{I_s^u, D_s, \alpha, \beta\}$ for slice $s \in {\mathcal S}^u$, where $I_s^u$ is the number of URLLC devices in $s$, $D_s$ is the transmission latency requirement of each URLLC device in $s$, $\alpha$ and $\beta$ are the packet blocking probability threshold and the codeword error decoding probability threshold of each URLLC device, respectively \cite{How2019yang}.
\end{mydef}

In this definition, URLLC devices are grouped into $|{\mathcal S}^u|$ clusters according to the transmission latency requirement of each device. As URLLC packets may arrive in burst and network resources allocated to URLLC slices may be inadequate, URLLC packets may experience blocking. The packet blocking probability threshold is then involved as a QoS requirement of URLLC slices.
Besides, owing to the low latency requirement, URLLC packets should be immediately scheduled upon arrival; thus, URLLC slice requests should always be accepted by the RAN-C {if there are sufficient network resources}. 

\section{Problem formulation}
Based on the system model, this section aims to formulate the problem of CoMP-enabled RAN slicing for mIoT and bursty URLLC multiplexing service provision.
{To this aim, we first present mIoT and bursty URLLC slice constraints and physical resource constraints. Then, we define slice utility functions of the problem. With the constraints and utility functions, the CoMP-enabled RAN slicing problem is formulated.}
{\subsection{mIoT slice constraint}}
For an IoT device in $s$, if it has the opportunity to send endogenous arrival packets to {its} corresponding RRH, then it will randomly select a preamble (e.g., orthogonal Zadoff-hu sequences) from a BBU-maintained preamble pool and transmit the preamble to the RRH. Just like the literature \cite{soorki2017stochastic,jiang2018random}, although the whole connection establishment process usually follows a RACH four-step procedure \cite{grau2019preamble}, we assume that a connection between the IoT device and the RRH is set up if the preamble can be successfully transmitted.
In other words, the RA success probability is regarded as the probability of successfully transmitting a preamble.

For an RRH, if its received preamble SINR is no less than a preset SINR threshold, then the preamble is considered to be successfully transmitted.
As in \cite{jiang2018random}, we do not investigate the well-investigated preamble collision issue.
{However,} owing to the channel deep fading and severe co-channel interference, an IoT device may experience uplink preamble transmission outage.
{Then,} at minislot $t$, for a randomly selected active IoT device in $s \in \mathcal{S}^I$, its RA success probability is defined as
\begin{equation}\label{eq:preamble_trans_suc_prob}
    P_{s}{(t)} = {\mathbb P}\{{ SINR}_s(t) \ge \theta_s^{th}\} \ge {\pi _s},
\end{equation}
where $SINR_s(t)$ denotes the preamble SINR experienced at an RRH associating with the IoT device, $\pi_s$ denotes a threshold of the required RA success probability.

We utilize a power-law path-loss model to calculate the path-loss between an IoT device and its RRH in mIoT slices. To eliminate the `near-far' effect, a truncated channel inversion power control scheme is also exploited. In the pass-loss model, the IoT device's transmit power decays at the rate of $r^{-{\varphi}}$ with $r$ representing the propagation distance and $\varphi$ denoting the path-loss exponent. In the power control scheme, IoT devices associated with the same RRH compensate for the path-loss to maintain that the average received signal power of an IoT device at the RRH equals a threshold $\rho_o$. Without loss of generality, the cutoff threshold $\rho_o$ is set to be the same for all RRHs,
and we perform the analysis of RA success probability on an RRH located at the origin. According to {the} Slivnyak's theorem \cite{haenggi2012stochastic}, the analysis holds for a generic RRH located at a generic location. For {the} randomly selected {active} IoT device in $s \in \mathcal{S}^I$, the preamble SINR experienced at the RRH located at the origin can take the form
\begin{equation}\label{eq:SINR}
    SINR_s(t) = {{\rho_o {h_o}}}/({{{\sigma ^2} + {{\mathcal I}_{s}}(t)}}),
\end{equation}
where $\sigma^2$ represents the noise power, ${{\mathcal I}_{s}}(t)$ denotes intra-cell interference received at the RRH{\footnote{{Just like \cite{zhang2015resource}, the co-channel inter-cell interference received by each RRH is assumed as a part of thermal noise mainly because of the intra-slice (or mIoT slice) interference isolation, the long-distance fading, and the severe wall penetration loss. Therefore, we focus on the analysis of the intra-cell interference in this paper.}}}, the useful signal power equals $\rho_o h_o$ due to the truncated channel inversion power control \cite{elsawy2014stochastic}, $h_o$ denotes the channel gain between the IoT device and the RRH. Note that the channel gain experienced at a generic RRH is related to the spatial locations of both the RRH and its associated IoT devices. Nevertheless, we drop the spatial indices for notation lightening.
Besides, just like \cite{elsawy2014stochastic}, all channel gains are assumed to be {independent of spatial locations and independent} and identically distributed (i.i.d.). Considering the particular IoT device deployment environment, the Rayleigh fading is assumed, and the channel gain is assumed to be exponentially distributed with unit mean \cite{elsawy2014stochastic}.
The intra-cell interference received at the origin RRH can take the following form
\begin{equation}\label{eq:I_intra}
\begin{array}{*{20}{l}}
    {\mathcal I}_{s}(t) = \sum\nolimits_{m \in {{ u}_{s}^I} \backslash \{ o\} } {\mathbbm 1}({p_m}||{d_m}|{|^{ - \varphi }} = {\rho _o}) \times \\
    \qquad {\mathbbm  1}(N_{a,s}{(t)} > 0){{\mathbbm 1}({f_m} = {f_o}){\rho _o}{h_m}},
\end{array}
\end{equation}
where ${ u}_s^I$ is the set of IoT devices connecting to the origin RRH in $s \in {\mathcal S}^I$, $o$ is the randomly selected IoT device associated with the RRH at the origin, $p_m$ denotes the transmit power of the $m$-th IoT device, $||d_m||$ is the distance between the $m$-th IoT device and the origin RRH, {$N_{a,s}(t)$ is the \underline{\textbf a}ccumulated number of packets in a queue during $t$, which is maintained by the selected IoT device in slice $s$ for packet behavior (e.g., arrival and departure) analysis,} $f_o$ denotes the preamble and channel chosen by the randomly selected IoT device, $h_m$ is the channel gain from the $m$-th IoT device to the origin RRH. $f_o = f_m$ indicates that the randomly selected IoT device and the $m$-th IoT device select the same preamble and channel.
The $1^{\rm st}$ ${\mathbbm 1}(\cdot)$ (from left to right) on the right-hand-side of (\ref{eq:I_intra}) indicates that the average received signal power of an interfering device at the origin RRH equals $\rho_o$ {owing to the adoption of the  truncated channel inversion power control scheme}. The $2^{\rm nd}$ ${\mathbbm 1}(\cdot)$ denotes that {an interfering device must be active}. The $3^{\rm rd}$ ${\mathbbm 1}(\cdot)$ indicates that an interfering device selects the same preamble and channel as the randomly selected IoT device.

{\subsection{Bursty URLLC slice constraint}}


During minislot $t$, a \emph{compound} Poisson process \cite{becchi2008poisson}, where arrivals happen in bursts (or batches, i.e., several arrivals can happen at the same instant) and the inter-batch times are independent and exponentially distributed, is utilized to model the number of bursty URLLC packets arrive at each RRH. The intensity of the exponential distribution is set to be one batch. The number of new arrivals in each batch is subject to an independent homogeneous Poisson distribution with intensity ${\bm \lambda} = \{\lambda_s; s \in {\mathcal{S}^u}\}$, where $\lambda_s$ denotes the intensity of new arrivals in a batch destined to devices belonging to URLLC slice $s$.
Once arrived, new URLLC arrivals will enter a queue maintained by an RRH to be {served}. An $M/M/W^u$ queueing system with limited bandwidth $W^u$ is exploited to model the queue due to the low latency requirement.
Without loss of any generality, we assume that each RRH maintains the same queue due to the exploration of cooperated transmission. In the queue, a packet destined to URLLC device $i \in {\mathcal{I}_s^u}$, $s \in {\mathcal{S}^u}$ will be allocated with a block of system bandwidth $\omega_{i,s}^u(t)$ for a period of time $d_s \le D_s$ at minislot $t$. Owing to stochastic variations in the bursty packet arrival process, the limited bandwidth may not be enough to serve new arrivals occasionally. As such, URLLC packet blocking may happen. To reduce the probability of URLLC packet blocking,
the PRB in the frequency domain for URLLC should be narrowed while widening it in the time domain \cite{anand2018resource}. In this way, the number of concurrent transmissions will be increased, and the packet blocking probability is reduced. As the tolerable communication latency of a URLLC device $i \in {\mathcal I}_s^u$ in {slice} $s \in \mathcal{S}^u$ is $D_s$, we can scale up $d_s$ and choose $d_s$ and $\omega_{i,s}^u(t)$ at minislot $t$ using the following equation
\begin{equation}\label{eq:omega_s}
   d_{{s}}  = D_s \ {\rm and} \ \omega_{{i,s}}^{u}(t) = {b_{i,s}^u(t){r_{{i,s}}^{u}(t)}}/{{(\kappa D_s)}},
\end{equation}
where $r_{i,s}^u(t)$ denotes channel uses for transmitting a URLLC packet \cite{anand2018resource}, $\kappa$ is a constant reflecting the number of channel uses per unit time per unit bandwidth of FDMA frame structure and numerology, $b_{i,s}^u(t) \in \{0, 1\}$ is an indicator that indicates whether the device $i$ in $s$ can be served at minislot $t$. As mentioned above, because network resources are limited and shared by all network slices, not all URLLC devices can be guaranteed to be served at every minislot although the RAN-C will always accept the URLLC slice requests. If the QoS requirement of $i$ in $s$ is satisfied at $t$, then the device $i$ can be served by the slice $s$, and we let $b_{i,s}^u(t) = 1$; otherwise, $i$ cannot be served by $s$, and we let $b_{i,s}^u(t) = 0$. Certainly, we can adjust the slice priority weight introduced in subsection III.D to {orchestrate network} resources for the {coverage of all} URLLC devices.

Based on the result in (\ref{eq:omega_s}), at minislot $t$, for a given $M/M/W^u$ queue with packet arrival intensity $\bm \lambda$, the minimum upper bound of bandwidth orchestrated for URLLC slices with a packet blocking probability $\alpha$ and a packet queueing probability $\varsigma$ can be given by \cite{How2019yang}
    \begin{equation}\label{eq:URLLC_bandwidth}
        \begin{array}{l}
            W^u(\bm r(t)) \le \sum\limits_{s \in {\mathcal S}^u } {\sum\limits_{i \in {\mathcal I}_s^u} {{\lambda _{s}}b_{i,s}^u(t)\frac{{{{r_{i,s}^u(t)}}}}{\kappa}} } + \\
            \frac{{{\alpha} - \varsigma {\alpha}}}{{\varsigma  - {\alpha}}}\sqrt {\frac{({\sum\limits_{s \in {{\mathcal S}^u}} {\sum\limits_{i \in {\mathcal I}_s^u}{b_{i,s}^u(t)\lambda _s^2D_s^2}} })( {\sum\limits_{s \in {\mathcal S}^u} {\sum\limits_{i \in {\mathcal I}_s^u} {{\lambda _{s}}b_{i,s}^u(t)\frac{{r_{i,s}^u(t)}^2}{{{\kappa ^2}{D_s}}}} }})}{{\mathop {\min }\nolimits_{s \in {{\mathcal S}^u}} \{ {\lambda _s}{D_s}\} }}}.
        \end{array}
    \end{equation}

As (\ref{eq:URLLC_bandwidth}) is correlated with the channel use $r_{i,s}^u (t)$, we next discuss how to obtain its expression.
For any URLLC slice $s \in {\mathcal S}^u$, during minislot $t$, let $\bm g_{ij,s}(t) \in {\mathbb C}^K$ be the transmit beamformer pointing at the device $i$ from the $j$-th RRH and $\bm h_{ij,s}(t) \in \mathbb{C}^K$ be the channel coefficient between the $i$-th URLLC device and the $j$-th RRH.
{Recall that} RRHs cooperate to transmit signals to a URLLC device to satisfy its reliability requirement, the signal-to-noise ratio (SNR) received at device $i$ in $s$ at minislot $t$ can be written as
\begin{equation}\label{eq:URLLC_SINR}
    {SNR_{i,s}^u(t)} = {{|\sum\nolimits_{j \in {\mathcal J}} {{{\bm h}_{ij,s}^{\rm H}}{{(t)}}{{\bm g}_{ij,s}}(t)} {|^2}}}/{{\phi \sigma _{i,s}^2}},
\end{equation}
where ${\mathcal J} = \{1,2,\ldots,J\}$ denotes the set of deployed RRHs, $\phi > 1$ is an SNR loss coefficient owing to imperfect channel status information acquisition \cite{hou2018burstiness}, $\sigma_{i,s}^2$ denotes the noise power. Just like \cite{tang2019service}, (\ref{eq:URLLC_SINR}) does not include interference due to the usage of a flexible FDMA mechanism.

For URLLC transmission {where the short packet transmission scheme is leveraged}, the capacity analysis for a finite blocklength channel coding regime derived in \cite{yang2014quasi} shall be resorted. {However, the capacity formula in \cite{yang2014quasi} was {derived under} an additive white Gaussian noise (AWGN) channel assumption rather than a fading channel {assumption} \cite{yang2014quasi,Durisi2016toward}. To tackle this issue, like \cite{Yang2020Joint,Ren2020Joint}, we assume that the fading channel is a quasi-static Rayleigh fading channel over a minislot and the channel coefficients are i.i.d. Then,}
for any device $i \in {\mathcal{I}_s^u}$, $s \in \mathcal{S}^u$, the number of transmitted information bits $L_{i,s}^u(t)$ at minislot $t$ using $r_{i,s}^u(t)$ channel uses can be approximately correlated with the codeword\footnote{It is noteworthy that a URLLC packet will usually be coded before transmission. The generated codeword will be transmitted in the air interface such that the transmission reliability can be improved.} error decoding probability $\beta$
\begin{equation}\label{eq:URLLC_bit_length}
\begin{array}{l}
        L_{i,s}^u(t)  \approx  r_{i,s}^u(t)C(SNR_{i,s}^u(t)) - \\
        \qquad {Q^{ - 1}}(\beta)\sqrt {r_{i,s}^u(t)V(SNR_{i,s}^u(t))},
\end{array}
\end{equation}
where $C(SNR_{i,s}^u(t)) = \log_2(1 + SNR_{i,s}^u(t) )$,
$V(SNR_{i,s}^u(t)) = \ln^2 2\left( {1 - \frac{1}{{{{(1 + SNR_{i,s}^u(t))}^2}}}} \right)$ is the channel dispersion, and $Q(\cdot)$ is the $Q$-function.

The complicated expression of $V(SNR_{i,s}^u(t))$ in (\ref{eq:URLLC_bit_length}) significantly hinders the theoretical analysis of network resources orchestrated for URLLC slices.
Fortunately, as $V(SNR_{i,s}^u(t))$ is upper-bounded by $\ln^2 2$, we can obtain the closed-form expression of the minimum upper bound of $r_{i,s}^u(t)$. {Specifically,} by defining $x =  \sqrt{r_{i,s}^u(t)}$ and solving a quadratic equation with respect to (w.r.t.) $x$, {the closed-form expression of $r_{i,s}^u(t)$ can take the following form}
\begin{equation}\label{eq:URLLC_channel_use}
\begin{array}{l}
r_{i,s}^u(t) \le \frac{{L_{i,s}^u(t)}}{{C(SNR_{i,s}^u(t))}} + \frac{{{{({Q^{ - 1}}(\beta))}^2}}}{{2{{(C(SNR_{i,s}^u(t)))}^2}}} + \\
\quad \frac{{{{({Q^{ - 1}}(\beta))}^2}}}{{2{{(C(SNR_{i,s}^u(t)))}^2}}}\sqrt {1 + \frac{{4L_{i,s}^u(t)C(SNR_{i,s}^u(t))}}{{{{({Q^{ - 1}}(\beta))}^2}}}}.
\end{array}
\end{equation}

\subsection{Physical resource constraints}
Next, we describe the physical resource constraints enforced for the RAN slicing system.

In mIoT slices, each RRH may transmit feedback signals to its connected IoT devices for the connection establishment according to the RACH four-step procedure \cite{grau2019preamble}. Meanwhile, in URLLC slices, each RRH may transmit URLLC packets to URLLC devices. As the transmit power $E_j$ ($j \in \mathcal{J}$) of each RRH is limited, we have the following transmit power constraint
\begin{equation}\label{eq:RRH_energy}
\begin{array}{l}
   \sum\nolimits_{s \in {\mathcal S}^I} {{(1 + \alpha_g)\frac{\lambda_s^I}{\lambda_R}{{\hat E}_{j}^I}}}  + \\
   \quad \sum\nolimits_{s \in {\mathcal S}^u} {\sum\nolimits_{i \in {\mathcal I}_s^u} {b_{i,s}^u(t){{\bm g}_{ij,s}^{\rm H}}{{(t)}}{{\bm g}_{ij,s}}(t)} }  \le {E_j},
\end{array}
\end{equation}
where ${\hat E}_{j}^I$ is assumed to be a constant and denotes the average transmit power of the $j$-th RRH for connecting to an associated IoT device over downlink, $\alpha_g$ is a coefficient. As a PPP with intensity $\lambda_s^I$ is utilized to model the distribution of IoT devices, the actual number of IoT devices may be greater than $\lambda_s^I$ once deployed. As a result, the coefficient $\alpha_g$ is introduced to reserve transmit power for exceeded IoT devices.

In the RAN slicing system, as the total limited system bandwidth $W$ will be shared by mIoT slices and URLLC slices, we have the following bandwidth constraint
\begin{equation}\label{eq:total_bandwidth}
   \sum\nolimits_{s \in {\mathcal S}^I} {(1+\alpha_g)\omega_s(\bar t)}  + W^u(\bm r(t))  \le W,
\end{equation}
where $\omega_s(\bar t)$ denotes the bandwidth allocated to mIoT slice $s \in \mathcal{S}^I$. {$\omega_s(\bar t)$ is correlated with $F_s$ by means of $F_s = \left\lfloor {{{\omega _s}(\bar t)}/{{a}}} \right\rfloor$. This is because $F_s$ orthogonal uplink physical RA channels (PRACHs) will be allocated to the mIoT slice $s$. Besides, a single tone mode \cite{LTE-Nokia} with the tone spacing of $a$ MHz is adopted for each uplink PRACH, which indicates that each PRACH occupies a PRB.} $\alpha_g \omega_s(\bar t)$ denotes a block of reserved bandwidth for exceeded IoT devices.

In (\ref{eq:total_bandwidth}), $F_s$ is an integer, and some integer variable recovery schemes \cite{tang2019systematic} can be leveraged to obtain the suboptimal $F_s$. However, considering the high computational complexity of optimizing an integer variable and the utilization of the scheme of reserving additional bandwidth resources, we directly relax the integer variable into a continuous one, i.e., let $F_s = {{{\omega _s}(\bar t)}/{{a}}}$. Without loss of any generality, we regard $\omega_s(\bar t)$ as an independent variable below.
Additionally, as at least one PRB should be allocated to each type of mIoT slice $s \in {\mathcal S}^I$, we have
\begin{equation}\label{eq:mMTC_bandwidth}
    \omega_s(\bar t) \ge a.
\end{equation}

\subsection{Slice utility functions}
Owing to the exploration of mIoT and bursty URLLC service multiplexing, we should orchestrate network resources for all mIoT slices and URLLC slices to simultaneously maximize the slice utilities.
For a mIoT slice $s \in \mathcal{S}^I$, its primary goal is to offload as many packets as possible from IoT devices. Thus, the number of accumulated packets in each IoT device should be kept at a low level.
Considering that a great RA success probability of an IoT device will lead to a low number of accumulated packets, we define the mIoT slice utility as follows.
\begin{mydef}\label{mydef:IoT_slice_utility}
Over a time slot of duration $T$, the mIoT slice utility is defined as the time average of RA success probabilities of IoT devices in all mIoT slices, which is given by
\begin{equation}\label{eq:IoT_utility}
{{\bar U}^I} = \frac{1}{T}\sum\nolimits_{t = 1}^T {{U^I}(t)}  = \frac{1}{T}\sum\nolimits_{t = 1}^T {\tilde P(t) },
\end{equation}
where $\tilde P(t) = \sum\nolimits_{s \in {{\mathcal S}^I}} {\frac{{{\lambda _s^I}{P_s}(t)}}{{\sum\nolimits_{s \in {{\mathcal S}^I}} {{\lambda _s^I}} }}}$ with the numerator $\lambda_s^I P_s (t)$ represents the expected sum of RA success probabilities of IoT devices in slice $s \in \mathcal{S}^I$ and the denominator $\sum\nolimits_{s\in\mathcal{S}^I}{\lambda_s^I}$ denoting a normalization coefficient.
\end{mydef}

In (\ref{eq:IoT_utility}), ${\lambda_s^I}/{\sum\nolimits_{s \in \mathcal{S}^I} {\lambda_s^I}}$ can be regarded as an intra-slice priority coefficient. A mIoT slice serving more IoT devices will be orchestrated with more network resources.

For a URLLC slice $s \in \mathcal{S}^u$, its primary objective is to maximize the slice gain reflected by the parameters of the bursty URLLC slice request {in an efficient way}. Therefore, we define an energy-efficient utility for URLLC slices, as presented below.
\begin{mydef}\label{mydef:URLLC_slice_utility}
Over one time slot of duration $T$, the bursty URLLC slice utility is defined as the time-average energy efficiency for serving URLLC devices, which is given by
\begin{equation}\label{URLLC_long_term_utility}
    \begin{array}{l}
        {{\bar U}^u} = \frac{1}{T}\sum\limits_{t = 1}^T {{U^u}(t)}  = \frac{1}{T}\sum\limits_{t = 1}^T {\sum\limits_{s \in {{\mathcal S}^u}} {U_s^u({D_s},{\bm g_{ij,s}}(t))} } \\
        \quad = \frac{1}{T}\sum\limits_{t = 1}^T {\sum\limits_{s \in {{\mathcal S}^u}} {\sum\limits_{i \in {\mathcal I}_s^u} \frac{b_{i,s}^u(t)}{{1 - {e^{ - {D_s}}}}} } }  - \\
        \qquad \frac{\eta }{T}\sum\limits_{t = 1}^T {\sum\limits_{s \in {{\mathcal S}^u}} {\sum\limits_{j \in {\mathcal J}} {\sum\limits_{i \in {\mathcal I}_s^u} {b_{i,s}^u(t)\bm g_{ij,s}^{\rm{H}}(t){\bm g_{ij,s}}(t)} } } },
   \end{array}
\end{equation}
where $\eta$ is {a positive} energy efficiency coefficient indicating the tradeoff between the URLLC slice gain and the RRH power consumption.
\end{mydef}

In (\ref{URLLC_long_term_utility}), we characterize the slice gain by $\frac{1}{T}\sum\nolimits_{t = 1}^T {\sum\nolimits_{s \in {{\mathcal S}^u}} {\sum\nolimits_{i \in {\mathcal I}_s^u} \frac{b_{i,s}^u(t)}{{1 - {e^{ - {D_s}}}}} } }$ as it reflects the latency requirements of bursty URLLC slices.
Then, during a time slot, the original RAN slicing problem for mIoT and URLLC service multiplexing can be formulated as follows.
\begin{subequations}\label{eq:original_problem}
\begin{alignat}{2}
& \mathop {\rm maximize}\limits_{\{{b_{i,s}^u(t)},{{\omega_s}}(\bar t),{{\bm g}_{ij,s}}(t)\}} {{\bar U}^I}  + {\tilde \rho} {{\bar U}^u} \\
& {\rm s.t. \text{ }}  b_{i,s}^u(t) \in \{ 0,1\}, \forall s \in {\mathcal S}^u, i \in {\mathcal I}_s^u \\
& \rm {constraints \text{ } (\ref{eq:preamble_trans_suc_prob}), (\ref{eq:RRH_energy})-(\ref{eq:mMTC_bandwidth}) \text{ } are \text{ } satisfied,}
\end{alignat}
\end{subequations}
where $\tilde \rho$ is {a non-negative} inter-slice priority coefficient reflecting the priority of orchestrating network resources for mIoT slices and URLLC slices\footnote{{To make the problem (\ref{eq:original_problem}) slightly simpler}, we do not focus on the selection of the optimal {inter-slice priority} coefficient $\tilde \rho $ and {energy efficiency coefficient} $\eta$ here. Yet, an iterative method proposed in our paper \cite{xi2020non} can be leveraged to determine their values.}.

The solution of (\ref{eq:original_problem}) is quite challenging mainly because i) \textbf{indeterministic objective function}: {the closed-form expression of $P_s(t)$ is not obtained. Besides,} (\ref{eq:original_problem}) should be optimized at the beginning of the $1^{\rm st}$ minislot. The time-averaged objective function of (\ref{eq:original_problem}) can only be exactly computed according to the future channel information. Therefore, the value of the objective function is indeterministic at the beginning of the $1^{\rm st}$ minislot; ii) \textbf{two-timescale issue}: the creation of a network slice is performed at a timescale of time slot. Thus, the variable $\omega_s(\bar t)$ should be determined at the beginning of the time slot $\bar t$ and kept unchanged over the whole time slot. The channel, however, is time-varying. As a result, the beamformer $\bm g_{ij,s}(t)$ should be optimized at each minislot $t$. In summary, the variables in (\ref{eq:original_problem}) should be optimized at two different timescales; iii) \textbf{thorny optimization problem}: at each minislot $t$, the constraint (\ref{eq:preamble_trans_suc_prob}) is non-convex over $\omega_s(\bar t)$, and the constraints (\ref{eq:RRH_energy}), (\ref{eq:total_bandwidth}) are non-convex over $\bm g_{ij,s}(t)$, which together lead to a non-convex problem.

{To solve this highly challenging problem, we first derive the closed-form expression of $P_s(t)$. Next, we attempt to tackle the two-timescale issue via converting it to single-timescale issues. Finally, we develop a novel alternative optimization method to solve the thorny optimization problem. The procedures of solving (\ref{eq:original_problem}) are elaborated in the following sections.}


{\section{Derivation of the closed-form expression of the RA success probability}}
{The RA success probability $P_s(t)$ of an IoT device is closely related to whether the device needs to request for the RACH and whether the RACH request is restricted. If the device has IoT packets to deliver, the device will request for the RACH. Therefore, we analyze behaviors (i.e., arrival, accumulation, and departure) of IoT packets in an IoT device and the probability of restricting its RACH request. Based on the analysis results, the closed-form expression of $P_s(t)$ is derived.}

{\subsection{Arrival, accumulation and departure of IoT packets}}
For a typical IoT device, we leverage a queue maintained in the device to capture the arrival, accumulation and departure of IoT packets.
During minislot $t$, a Poisson distribution with intensity (or {the average number of} ne{\underline{\textbf w}} arrival {packets}) {$\vartheta_{w,s}(t)$} is exploited to model the random, mutually independent endogenous packet arrivals in an IoT device in slice $s$.
{Once arrived, new packets will not be sent out immediately in general and will enter a queue in the IoT device to wait to be served. To model the queue, an $M/M/k$ queueing system rather than an $M/M/{1}$ queueing system is leveraged as the key performance indicators of the former are much better than the later one. First-come-first-serve (FCFS) is selected as the queueing principle in the $M/M/k$ queue.}
Besides, to facilitate the analysis of the queue evolution process, we consider the slotted-ALOHA RA protocol although there are many other RA protocols such as non-orthogonal and coded RA protocols. Owing to the RA behavior of the ALOHA protocol, new arrivals during $t$ will only be counted at minislot $t + 1$. Thus, the {value of} $N_{a,s}(t)$ in the queue of a randomly selected IoT device in slice $s$ at $t$ is simultaneously determined by the following three factors: a) the accumulated number of packets; b) the number of new arrivals during $t - 1$; c) whether the preamble of the device can be successfully decoded by its associated RRH. The work in \cite{jiang2018random} presented a queue evolution model based on the single packet transmission configuration. We extend \cite{jiang2018random} to the {general} case of {transmitting multiple packets in one transmission time interval} {as multiple packets can be simultaneously transmitted in one transmission time interval},
and (\ref{eq:queue_evolution}) shows an evolution model of $N_{a,s}(t)$ for all $s \in {\mathcal S}^I$ with
\begin{equation}\label{eq:queue_evolution}
{N_{a,s}}(t) = \left\{ {\begin{array}{*{20}{l}}
{0,t = 1}\\
{{{[{N_{w,s}}(t - 1) - \mathbbm 1({\rm{RA \text{ } succeeds}}){x_s}]}^ + },t = 2}\\
{[{N_{a,s}}(t - 1) + {N_{w,s}}(t - 1) - }\\
{\mathbbm 1({\rm{RA \text{ } succeeds}}){x_s}{]^ + },t \ge 3}
\end{array}} \right.
\end{equation}
where $N_{w,s}(1)$ is the number of ne{\underline{\textbf w}} arrivals in the $1^{\rm st}$ minislot, ${\mathbbm 1}(\cdot)$ is a function equaling one if the corresponding RA succeeds; otherwise, ${\mathbbm 1}(\cdot)=0$. $x_s = a\log_2(1+\theta_s^{th})/L$ packets at the head of the queue will be popped out if ${\mathbbm 1}(\cdot)=1$, where $L$ is the {number of information bits of an} IoT packet; otherwise, they will not. $[x]^+ = \max(x,0)$.

At minislot $t$, based on the model in (\ref{eq:queue_evolution}), for a randomly selected IoT device in slice $s \in \mathcal{S}^I$, the probability that its maintained queue is {\underline{\textbf n}}ot {\underline{\textbf e}}mpty can be defined as
\begin{equation}\label{eq:non_empty_prob}
    P_{ne,s}{(t)} = {\mathbb P}\{N_{a,s}(t) > 0\}.
\end{equation}

(\ref{eq:non_empty_prob}) explicitly shows that new arrivals at $t$ will not be sent out immediately, which is reflected in (\ref{eq:queue_evolution}). (\ref{eq:non_empty_prob}) is significantly different from the work in \cite{jiang2018random}, which defines the non-empty probability ${\mathcal T}^m = {\mathbb P}\{N_{\rm Cum}^m + N_{\rm New}^m> 0 \}$, where $N_{\rm Cum}^m$ is the number of accumulated packets and $N_{\rm New}^m$ denotes the number of new arrivals in the $m$-th slot. The definition ${\mathcal T}^m$ shows that new arrivals during the $m$-th minislot have the probability of sending out immediately. {However, according to the RA behavior of the ALOHA protocol, new arrivals in minislot $m$ can only be sent out in minislot $m+1$ if possible.}


Next, we describe the packet departure process combined with {an extended subframe structure for mIoT transmissions}.
As mentioned above, partly due to the limitation on the subframe structure, NB-IoT and LTE-M {\cite{Cellular-Ericsson,LTE-Nokia}} can only admit 50,000 devices.
For NB-IoT, only one PRB with a bandwidth of $180$ KHz in the frequency domain is allocated for the IoT service, and each physical channel occupies the whole PRB.
For LTE-M, although the physical channels are time and frequency multiplexed, it only reserves six in-band PRBs with a total bandwidth of $1.08$ MHz in the frequency domain for the IoT service.
Thus, the subframe structure for mIoT transmissions should be revisited if more RACH requests from IoT devices are required to be accepted.

Fig. \ref{fig:fig_frame_minislot_structure} depicts an extended subframe structure for mIoT transmissions. Although it depicts some essential channels, we do not discuss their correlations to the considered {RAN slicing} problem as the detailed research on the physical layer supporting the mIoT service is out of our scope.
In this structure, both the frequency division multiplexing (FDM) scheme and {the} code division multiplexing (CDM) scheme are leveraged to admit more IoT devices in the way of alleviating mutual device interference. Particularly, the FDM scheme can alleviate signal interference through orthogonal frequency band allocation. The CDM scheme mitigates the co-channel signal interference via reducing the cross-correlation of simultaneous transmissions.
{Based on the extended structure}, at the beginning of each minislot, an active IoT device, i.e., an IoT device whose queue is non-empty, will randomly choose a preamble, {which is an orthogonal sequence,} from a set of non-dedicated RA preambles of size $\xi$. {Next, it will} transmit the preamble through a randomly selected PRACH, {which occupies a PRB}. For each preamble, it has an equal probability ${{1}/{\xi}}$ to be chosen by each IoT device. Similarly, each PRACH has an equal probability ${1}/{F_s}$ to be selected. Thus, the average number of IoT devices in mIoT slice $s \in \mathcal{S}^I$ choosing the same PRACH and the same preamble is ${\lambda_s^I}/{\xi F_s}$. Notably, a greater $\xi F_s$ may significantly reduce signal interference experienced at each RRH.

\begin{figure}[!t]
\centering
\includegraphics[width=3.4in]{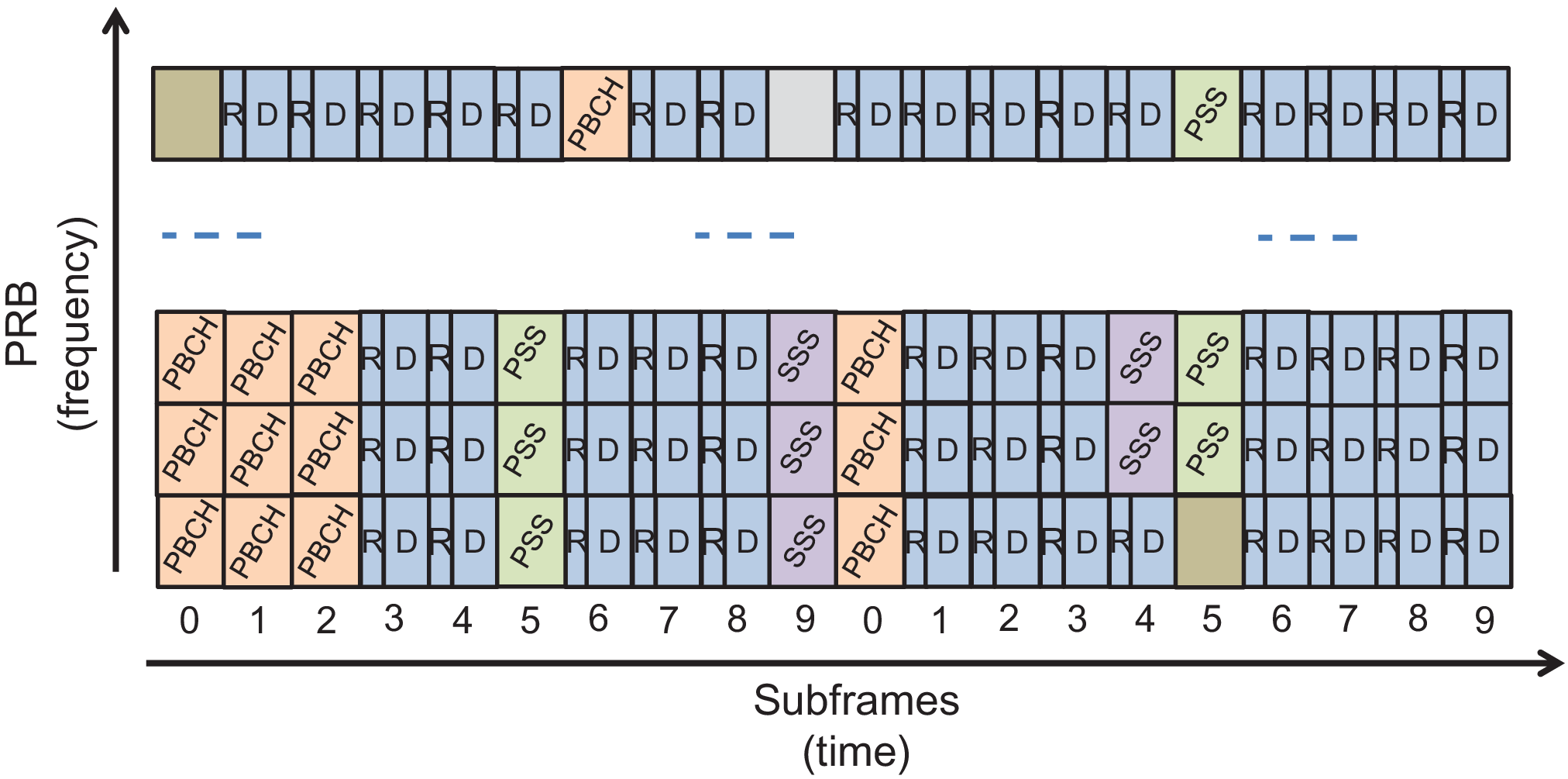}
\caption{{The extended subframe structure for mIoT transmissions.} 'R' and 'D' denote the resource block reserved for preamble and IoT data transmission. The preamble in 'R' also reflects the usage of a code division multiplexing scheme. PBCH, PSS and SSS represent the PRBs for physical broadcast channel, primary synchronization signal and secondary synchronization signal transmission, respectively.}
\label{fig:fig_frame_minislot_structure}
\end{figure}

{\subsection{Access control scheme}}
In a mIoT network slice, as the slotted-ALOHA protocol allows all active IoT devices to request for RA at the beginning of each minislot without checking channel statuses, IoT devices may simultaneously transmit preambles. It may incur severe slice congestion that may lower the RA success probabilities of IoT devices and degrade the system performance.
Access control has been considered as an efficient proposal of alleviating congestion \cite{Study163GPP}, and many access control schemes such as access class barring (ACB), power ramping and back-off schemes \cite{jiang2018random} have been proposed.
As we aim at investigating the performance difference between a network slicing system without access control and with access control, we adopt the following two schemes \cite{jiang2018random}: 1) \textbf{Unrestricted scheme:} each active IoT device requests the RACH at the beginning of minislot $t$ without access restriction;
2) \textbf{ACB scheme:} at the beginning of $t$, each active IoT device draws a random number $q \in [0, 1]$ and can request the RACH only when $q < P_{ACB}$, where $P_{ACB}$ is an ACB factor determined by RRHs based on the slice congestion condition.

With the introduced access control schemes, the probability that the RACH requests of a randomly selected IoT device in slice $s \in \mathcal{S}^I$ are {\underline{\textbf n}}ot {\underline{\textbf r}}estricted at minislot $t$ is defined as
\begin{equation}\label{eq:non_restriction_prob}
    P_{nr,s}{(t)} = {\mathbb P}\{{\rm Unrestricted} \text{ } {\rm RACH} \text{ } {\rm requests}\}.
\end{equation}

For all $s \in \mathcal{S}^I$ at any minislot $t$, we have $P_{nr,s}(t) = 1$ for the unrestricted scheme and $P_{nr,s}(t) = P_{ACB}$ for the ACB scheme.

{\subsection{Closed-form expression of $P_s(t)$}}
{With the above analysis results, we can now derive the closed-form expression of $P_s(t)$, {$\forall s \in {\mathcal S}^I$}. Specifically,} for the randomly selected IoT device in $s \in {\mathcal{S}^I}$, we can rewrite (\ref{eq:preamble_trans_suc_prob}) as follows 
\begin{equation}\label{eq:QoS_analysis}
    \begin{array}{l}
P_s(t) = {\mathbb P}\{ {h_o} \ge \frac{{\theta _s^{th}}}{{{\rho _o}}}({\sigma ^2} + {{\mathcal I}_{s}(t)})\} \\
\mathop  = \limits^{(a)} {{\mathbb E}}\left[ {\exp \left\{ { - \frac{{\theta _s^{th}}}{{{\rho _o}}}({\sigma ^2} + {{\mathcal I}_{s}(t)})} \right\}} \right] \\
=  e^ { { - \frac{{\theta _s^{th}}}{{{\rho _o}}}{\sigma ^2}} }{\mathcal L}_{{{\mathcal I}_{s}(t)}}\left(\frac{\theta_s^{th}}{\rho_o}\right),
\end{array}
\end{equation}
where (a) follows from the {law of total probability} over ${\mathcal I}_{s}(t)$, and ${\mathcal L}_{{{\mathcal I}_{s}(t)}} (\cdot)$ denotes the Laplace transform (LT) of the probability density function (PDF) of the random variable ${\mathcal I}_{s}(t)$. Note that the notation ${\mathcal L}_{{{\mathcal I}_{s}(t)}} (\cdot)$ is a terminology that is a slight abuse of subscript ${{\mathcal I}_{s}(t)}$.

The following lemma characterizes the LT of interference $\mathcal{I}_s(t)$. {By referring to the RA behavior of the ALOHA protocol, we derive the expression of the LT of interference $\mathcal{I}_{\rm Intra}(t)$, which is obviously different from that obtained in \cite{jiang2018random}.} In \cite{jiang2018random}, the obtained ${\mathcal L}_{{{\mathcal I}_{\rm Intra}(t)}}(\gamma_{th}/\rho)$ was a quasi-convex function over the system bandwidth allocated to IoT devices. In this paper, the obtained ${\mathcal L}_{{{\mathcal I}_{s}(t)}} \left(\theta_s^{th}/\rho_o\right)$ is the difference of two quasi-convex functions that significantly increases the difficulty of orchestrating system resources for mIoT slices.
\begin{lemma}\label{lem:LT_interference_expression}
For the origin RRH, based on the packet evolution model in (\ref{eq:queue_evolution}), the LT of its received interference from active IoT devices associated with it is given by
    \begin{equation}\label{eq:LT_interference_expression}
        {\mathcal L}_{{{\mathcal I}_{s}(t)}} \left(\varpi_s\right) = {\frac{{1 + \varpi_s{\rho _o}}}{{{{\left( {1 + \alpha_s \varpi_s{\rho _o}/\left( {1 + \varpi_s{\rho _o}} \right)} \right)}^{3.5}}}} - \frac{{1 + \varpi_s{\rho _o}}}{{{{\left( {1 + \alpha_s } \right)}^{3.5}}}}},
    \end{equation}
where $\varpi_s = {\theta_s^{th}}/{\rho_o}$, $\alpha_s  = {{{P_{nr,s}}(t){P_{ne,s}}(t){\lambda _{{s}}^I}}}/({{3.5{\lambda _R}{\xi F_s}}})$.
\end{lemma}
\begin{proof}
Please refer to Appendix A.
\end{proof}



By substituting (\ref{eq:LT_interference_expression}) into (\ref{eq:QoS_analysis}), we can obtain a mathematical expression of $P_s(t)$. The expression, however, is not in the closed-form as it is a function of $P_{ne,s}(t)$, the closed-form expression of which is not obtained. Next, we attempt to derive the closed-form expression of $P_{ne,s}(t)$.

According to the definition of non-empty probability, $P_{ne,s}(t)$ is correlated with $N_{a,s}(t)$. Thus, we theoretically analyze $P_{ne,s}(t)$ as the following.

From (\ref{eq:queue_evolution}), we can observe that $N_{a,s}(1) = 0$ for all $s \in {\mathcal S}^I$; thus, at the $1^{\rm st}$ minislot, we have
\begin{equation}\label{eq:non_empty_prob_1}
    P_{ne,s}^{1} = {\mathbb P}\{N_{a,s}^1 > 0\} = 0,
\end{equation}
where we write $x^t$ instead of $x(t)$ to lighten the notation. {The similar lightened notation is adopted throughout the rest of this section to simplify the description.}

The following lemma presents the closed-form expression of the non-empty probability of a randomly selected IoT device served by the origin RRH when minislot $t > 1$.
\begin{lemma}\label{lem:non_empty_prob_lemma}
    The number of accumulated packets of a randomly selected IoT device served by the origin RRH at minislot $t > 1$ may be approximately Poisson distributed. Therefore, based on the model in (\ref{eq:queue_evolution}), for any mIoT slice $s \in \mathcal{S}^I$, we approximate the number of {\underline{\textbf a}}ccumulated packets $N_{a,s}^{t}$ at $t$ as a Poisson distribution with intensity $\vartheta_{a,s}^{t}$, which is given by
    \begin{equation}\label{eq:mu_accumulated_packets}
        \vartheta _{a,s}^t = \left [\vartheta _{w,s}^{t - 1} + \vartheta _{a,s}^{t - 1} - P_{s}^{t - 1}\left( {1 - {e^{ - \vartheta _{w,s}^{t - 1} - \vartheta _{a,s}^{t - 1}}}} \right) \right ]^+
    \end{equation}

    Then, the probability that the queue of the device is non-empty at $t$ can be written as
    \begin{equation}\label{eq:non_empty_prob_m}
        P_{ne,s}^t = 1 - {e^{ - \vartheta _{a,s}^t}}
    \end{equation}
\end{lemma}
\begin{proof}
Please refer to Appendix B.
\end{proof}

Combined with (\ref{eq:QoS_analysis}), (\ref{eq:LT_interference_expression}) and (\ref{eq:non_empty_prob_m}), the closed-form expression of $P_s(t)$ ($s \in \mathcal{S}^I$) can be obtained.

\section{Problem solution with system generated channels}
{Although we obtain the closed-form expression of $P_s(t)$, it is still difficult to solve (\ref{eq:original_problem}). This is mainly because (\ref{eq:original_problem}) is a two-timescale optimization problem and some optimization methods cannot be directly utilized to solve it.
A possible proposal of solving the two-timescale optimization problem is to explore an SAA technique \cite{kim2015guide} and an ADMM method \cite{boyd2011distributed}.
The SAA technique can be utilized to approximate the indeterministic objective function. Based on the approximated results, the ADMM method can be leveraged to decompose the two-timescale problem into multiple single-timescale problems, which can be solved by some optimization methods. Then, the solution of the two-timescale problem can be effectively recovered by ADMM based on the solutions of the single-timescale problems.
}

\subsection{Sample average approximation and alternating direction method of multipliers}
As mIoT slices and URLLC slices share the network resources, both $\bar U^I$ and $\bar U^u$ may be determined by channel coefficients experienced by URLLC slices. At each minislot $t$, due to the i.i.d. assumption on the channel coefficients of URLLC slices, we have
\begin{equation}\label{objfun_approx}
    \frac{1}{T}\sum\nolimits_{t = 1}^T {{U^I}(t)}  + \frac{1}{T}\sum\nolimits_{t = 1}^T {\tilde \rho {U^u}(t)}  \approx {{\mathbb E}_{\hat {\bm h}}}\left[ {{{\hat U}^I} + \tilde \rho {{\hat U}^u}} \right]
\end{equation}
where $\hat {\bm h}$ is the channel samples of URLLC slices collected at the beginning of the time slot $\bar t$.

Given a collection of channel samples $\{\bm h_m\}$ with $\bm h_m = [\bm h_{11,1m};\ldots;\bm h_{1J,sm};\ldots;\bm h_{N^uJ,|\mathcal{S}^u|m}]$ and $m \in \mathcal{M}=\{1,\ldots,M\}$. For notation lightening, we write $x_m$ instead of $x(m)$ that represents a variable corresponding to the channel sample $\bm h_m$.
Just like \cite{How2019yang}, as constraints (\ref{eq:original_problem}b) and (\ref{eq:original_problem}c) construct a non-empty compact set, the conclusion of Proposition 5.1 in \cite{How2019yang} is applicable to this paper by exploiting the SAA technique.
The conclusion indicates that if the number of channel samples $M$ is reasonably large, then $\frac{1}{M}\sum\nolimits_{m = 1}^M {U_m^I }  + \frac{\tilde \rho}{M}\sum\nolimits_{m = 1}^M {U_m^u }$ converges to ${{\mathbb E}_{\hat {\bm h}}}[ {{{\hat U}^I} + {{\hat U}^u}} ]$ uniformly on the non-empty compact set almost surely. In other words, the SAA technique enables us to use the channel samples collected at the beginning of a time slot to approximate the unknown channel coefficients over the time slot.

Recall that the variable $\omega_s(\bar t)$ will be kept unchanged over the time slot $\bar t$ and the beamformer $\bm g_{ij,s}(t)$ should be calculated at each minislot $t$, we can further consider (\ref{eq:original_problem}) as a global consensus problem, which can be effectively mitigated by an ADMM method. In (\ref{eq:original_problem}), $\omega_s(\bar t)$ is a global consensus variable that should be maintained in consensus for all $\bm h_m$ ($m \in \mathcal{M}$), and $\bm g_{ij,sm}$ that is calculated based on $\bm h_m$ is a local variable.
The fundamental principle of ADMM is to impose augmented penalty terms characterizing global consensus constraints on the objective function of an optimization problem. In this way, the local variables can be driven into the global consensus while still attempting to maximize the objective function.
Let ${\bm G}_{i,sm} = {\bm g}_{i,sm} {\bm g}_{i,sm}^{\rm H} \in {\mathbb R}^{JK \times JK}$, ${\bm H}_{i,sm} = {\bm h}_{i,sm} {\bm h}_{i,sm}^{\rm H} \in {\mathbb R}^{JK \times JK}$, where ${\bm g}_{i,sm} = [{\bm g}_{i1,sm};\ldots;{\bm g}_{iJ,sm}] \in {\mathbb C}^{JK \times 1}$ and ${\bm h}_{i,sm} = [{\bm h}_{i1,sm};\ldots;{\bm h}_{iJ,sm}] \in {\mathbb C}^{JK \times 1}$. By applying the property \cite{karipidis2008quality} ${{{\bm G}_{i,sm}} = {{\bm g}_{i,sm}}{\bm g}_{i,sm}^{\rm H} \Leftrightarrow {{\bm G}_{i,sm}} \succeq 0}$, ${{\rm rank}({{\bm G}_{i,sm}}) \le 1}$ and utilizing the conclusions of SAA and ADMM, we can approximate (\ref{eq:original_problem}) as the following problem at the beginning of the time slot $\bar t$.
\begin{subequations}\label{eq:SAA_admm_problem}
\begin{alignat}{2} 
& \mathop {{\rm{minimize}}}\limits_{\{ \omega _{sm},\omega_s(\bar t),{b_{i,sm}^u},{\bm G_{i,sm}}\}\hfill}
\sum\limits_{m = 1}^M {\left[ { - \frac{{U_m^I}}{M} - \frac{{\tilde \rho U_m^u}}{M}} \right]}  + \nonumber \\
& \underbrace {\sum\limits_{m = 1}^M {\sum\limits_{s \in {{\mathcal S}^I}} {\left[ {{\psi _{sm}}\left( {{\omega _{sm}} - {\omega _s}(\bar t)} \right) + \frac{\mu }{2}\left\| {{\omega _{sm}} - \omega{_s}(\bar t)} \right\|_2^2} \right]} } }_{{\rm augmented} \text{ } {\rm penalty} \text{ } {\rm terms}} \allowdisplaybreaks[4] \\
& {\rm s.t. \text{ }}  P_{sm} \ge \pi_s, \forall s \in {\mathcal S}, m \in {\mathcal M} \\
& \sum\nolimits_{s \in {\mathcal S}^I} {{(1 + \alpha_g)\frac{\lambda_s^I}{\lambda_R}{\hat E}_{j}^I}}  + \nonumber \\
& \sum\nolimits_{s \in {{\mathcal S}^u}} {\sum\nolimits_{i \in {\mathcal I}_s^u} {b_{i,sm}^u {\rm tr}({{\bm Z}_j}{{\bm G}_{i,sm}})} }  \le {E_j}, \forall j\in {\mathcal J},m\in {\mathcal M} \\
& \sum\nolimits_{s \in {{\mathcal S}^I}} {(1+\alpha_g) \omega _{s}(\bar t)}  + W^u(\bm r_m)  \le W, m \in {\mathcal M} \\
&{{\bm G}_{i,sm}} \succeq 0, \forall s \in {\mathcal S}^u, i \in {\mathcal I}_s^u, m \in {\mathcal M}  \\
&{{\rm rank}({{\bm G}_{i,sm}}) \le 1}, \forall s \in {\mathcal S}^u, i \in {\mathcal I}_s^u, m \in {\mathcal M} \\
& b_{i,sm}^u \in \{0,1\}, \forall s \in {\mathcal S}^u, i \in {\mathcal I}_s^u, m \in {\mathcal M} \\
& {\rm constraint} \text{ } (\ref{eq:mMTC_bandwidth}) \text{ } {\rm is} \text{ } {\rm satisfied},
\end{alignat}
\end{subequations}
where $\psi_{sm}$ is the Lagrangian multiplier, $\mu$ is a penalty coefficient, ${\bm Z}_j$ is a square matrix with $J \times J$ blocks, and each block in ${\bm Z}_j$ is a $K \times K$ matrix. In ${\bm Z}_j$, the block in the $j$-th row and {the} $j$-th column is a $K \times K$ identity matrix, and all other blocks are zero matrices.

(\ref{eq:original_problem}) is now reduced to a deterministic single-timescale problem (\ref{eq:SAA_admm_problem}). What is more, (\ref{eq:SAA_admm_problem}) can be split into $M$ separate problems that can be optimized in parallel as its objective function is separable.
Thus, the following ADMM-based framework, {which needs to compute (\ref{eq:arg_lagarangian})-(\ref{eq:psi_update}),} can be exploited to mitigate (\ref{eq:SAA_admm_problem}).
\begin{subequations}\label{eq:arg_lagarangian}
\begin{alignat}{2}
& \left\{ {{\omega _{sm}^{(k + 1)},b_{i,sm}^{u(k+1)}},{\bm G_{i,sm}^{(k + 1)}}} \right\} = \nonumber \\
& \mathop {{\rm{argmin}}}\nolimits_{\left\{ \omega _{sm},b_{i,sm}^{u},{\bm G_{i,sm}}\right\}} \overline {\mathcal L} (\omega _{sm},{{\bm G}_{i,sm}})  \\
& {\rm s.t. \text{ }}  {\rm for} \text{ } {\rm the} \text{ }  m{\rm -th} \text{ } {\rm sample},  (\ref{eq:SAA_admm_problem}b)-(\ref{eq:SAA_admm_problem}g) \text{ } \rm{are} \text{ } {\rm satisfied} \allowdisplaybreaks[4] \\
& \quad {\rm for} \text{ } {\rm the} \text{ }  m{\rm -th} \text{ } {\rm sample}, \text{ } \omega_{sm} \ge a, \text{ } \forall s\in {\mathcal S}^I
\end{alignat}
\end{subequations}
\begin{equation}\label{eq:omega_update}
\omega _s^{(k + 1)}(\bar t) = \sum\nolimits_{m = 1}^M {( {\omega _{sm}^{(k + 1)} + \psi _{sm}^{(k)}/{\mu}} )}/{M}, \text{ } \forall s\in {\mathcal S}^I
\end{equation}
\begin{equation}\label{eq:psi_update}
    \psi _{sm}^{(k + 1)} = \psi _{sm}^{(k)} + \mu \left( {\omega _{sm}^{(k + 1)} - \omega _s^{(k + 1)}(\bar t)} \right), \text{ } \forall s\in {\mathcal S}^I
\end{equation}
where the augmented partial Lagrangian function
\begin{equation}\label{eq:average_Lagrangian_func}
\begin{array}{l}
    \bar{\mathcal{L}}(\omega _{sm},\bm {\bm G}_{i,sm}) =  {- \frac{{U_m^{I(k)}}}{M} - \frac{{\tilde \rho U_m^{u(k)}}}{M} +} \\
    { \sum\limits_{s \in {{\mathcal S}^I}} {\left[ {\psi _{sm}^{(k)}\left( {\omega _{sm} - \omega _s^{(k)}(\bar t)} \right) + \frac{\mu }{2}{{\left\| {\omega _{sm} - \omega _s^{(k)}(\bar t)} \right\|}_2^2}} \right]} }.
\end{array}
\end{equation}

This ADMM-based framework can be executed on multiple processors{/virtual machines}. Each processor is responsible for optimizing (\ref{eq:arg_lagarangian}) and calculating (\ref{eq:psi_update}) with a global value as an input. (\ref{eq:omega_update}) is centrally updated in such a way that local variables converge to the global value, which is the solution of (\ref{eq:SAA_admm_problem}).
Unfortunately, (\ref{eq:arg_lagarangian}) is a mixed-integer non-convex optimization problem as there are zero-one variables, continuous variables and non-convex constraints in (\ref{eq:arg_lagarangian}). As a result, the optimization of (\ref{eq:arg_lagarangian}) is quite difficult. We next discuss how to handle this hard problem.

\subsection{Alternative optimization}
In this subsection, we {explore a novel} alternative optimization scheme to handle the mixed-integer non-convex optimization problem. Specifically, we first assume that continuous variables are known and attempt to mitigate a zero-one optimization problem. Given the zero-one variables, we then try to optimize a non-convex optimization problem. The process is alternatively conducted until convergence.

\subsubsection{URLLC device associations}
Given continuous variables $\{{\bm G}_{i,sm}^{(k)}, \omega_{sm}^{(k)}\}$ at the $k$-th iteration, the association problem of URLLC devices in URLLC slices can take the following form
\begin{subequations}\label{eq:arg_lagarangian_b_isu}
\begin{alignat}{2}
& \{ b_{i,sm}^{u(k+1)} \} = \mathop {{\rm{argmin}}}\nolimits_{\{b_{i,sm}^u\}} - {{\tilde \rho U_m^{u(k)}}}/{M}  \\
& {\rm s.t. \text{ }}  {\rm for} \text{ } m, \text{ } (\ref{eq:SAA_admm_problem}c), (\ref{eq:SAA_admm_problem}d),(\ref{eq:SAA_admm_problem}g) \text{ } \rm{are} \text{ } satisfied.
\end{alignat}
\end{subequations}

This problem is non-linear and hard to be handled. In theory, an exhaustive algorithm can obtain the optimal solution of (\ref{eq:arg_lagarangian_b_isu}). The computation complexity of this algorithm is $O(2^{N^u})$ that may be impractical in implementation.
Therefore, a greedy scheme of the computational complexity $O(N^u)$, which is summarized as the following, is proposed to obtain $\{b_{i,sm}^{u(k+1)}\}$. 

\begin{enumerate}[a)]
    \item initialize two device sets, i.e., candidate device set ${\mathcal{I}}^{u-} = {\mathcal I}^u$, association device set ${\mathcal I}^{u+} = \emptyset$.
    \item select the device that maximizes ${{\tilde \rho U_m^{u(k)}}}/{M}$ from ${\mathcal{I}}^{u-}$, remove it from ${\mathcal{I}}^{u-}$, and add it to ${\mathcal{I}}^{u+}$. Given ${\mathcal{I}}^{u+}$, check the feasibility of (\ref{eq:arg_lagarangian_b_isu}). If (\ref{eq:arg_lagarangian_b_isu}) is feasible, then accept the device; otherwise, remove the device from ${\mathcal{I}}^{u+}$. Continue till ${\mathcal{I}}^{u-} = \emptyset$.
\end{enumerate}

\subsubsection{{Non-convex optimization}}
Given the obtained $b_{i,sm}^{u(k+1)}$, (\ref{eq:SAA_admm_problem}) will be reduced to the following {optimization} problem.
\begin{subequations}\label{eq:bandwidth_beamforming_problem}
\begin{alignat}{2}
& \left\{
{\omega _{sm}^{(k + 1)}},
{\bm G_{i,sm}^{(k + 1)}}
 \right\} = \mathop {{\rm{argmin}}}\nolimits_{\{\omega _{sm},{\bm G_{i,sm}}\}} \overline {\mathcal L} (\omega _{sm},{{\bm G}_{i,sm}})  \\
& {\rm s.t. \text{ }} {\rm for} \text{ } m,  (\ref{eq:SAA_admm_problem}b)-(\ref{eq:SAA_admm_problem}f),(\ref{eq:arg_lagarangian}c) \text{ } \rm{are} \text{ } satisfied.
\end{alignat}
\end{subequations}

In (\ref{eq:bandwidth_beamforming_problem}), the low-rank constraint (\ref{eq:SAA_admm_problem}f) is non-convex, and its objective function is not convex and even not quasi-convex w.r.t. ${\omega_{sm}}$, the tackling of which is quite tricky.
To tackle the non-convex low-rank constraint (\ref{eq:SAA_admm_problem}e), we resort to the SDR technique. The primary procedures of SDR are i) directly drop the low-rank constraint; ii) solve the optimization problem without the low-rank constraint to obtain the solution; iii) if the obtained solution is not rank-one, then some manipulations such as randomization/scale \cite{ma2010semidefinite} are needed to perform on it to impose the low-rank constraint; otherwise, its principal component is the optimal solution to (\ref{eq:bandwidth_beamforming_problem}).

For the tricky objective function, we are reminded of the art of dealing with a non-convex function, i.e., study the structure of the function if it is non-convex. A crucial observation is that $P_{sm}$ is quasi-concave w.r.t. $\omega_{sm}$ although the objective function is not quasi-convex w.r.t. $\omega_{sm}$. Therefore, we resort to the Taylor expansion to approximate the tricky objection function.

The following analysis is based on two facts \textbf{Fact 1:} the value of the objective function of (\ref{eq:bandwidth_beamforming_problem}) is mainly determined by that of ${\tilde P_{m}^{(k)}}$ (or $U_m^{I(k)}$); \textbf{Fact 2:} the solution $\omega_{sm}$ maximizing ${\tilde P_{m}^{(k)}}$ must locate in the range of $[\hat \omega_{sm}^{lb}, S_{sm}^{\star}]$ $\forall s,m$, as shown in Fig. \ref{fig:fig_quasi_convex_structure}, where $\hat \omega_{sm}^{lb} = \max \{\omega_{sm}^{lb}, a\}$ with $\omega_{sm}^{lb}$ denoting the lower bound of $\omega_{sm}$ satisfying the constraint (\ref{eq:SAA_admm_problem}b), $S_{sm}^{\star}$ is the $\omega_{sm}$ maximizing $P_{sm}$, and the notation $P_{sm}|_{\omega_{sm}}$ is utilized to explicitly indicate that $P_{sm}$ is a function of $\omega_{sm}$.

\begin{figure}[!t]
\centering
\subfigure[Curve of $P_{sm}$.]{\includegraphics[width=2.4 in]{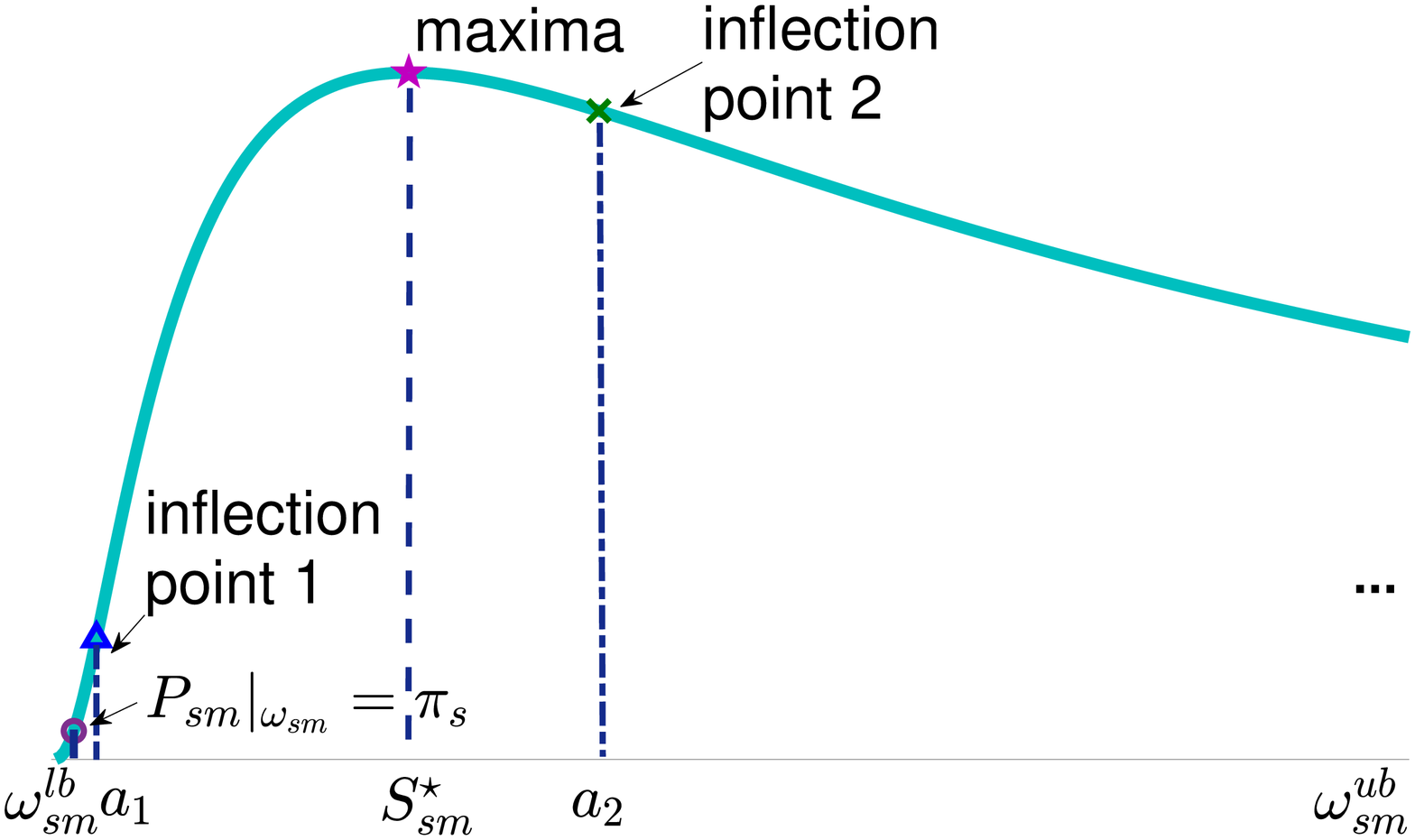}%
\label{fig:fig_curve_structure}}
\hspace{0.05\linewidth}
\subfigure[$2^{\rm nd}$ Taylor expansion.]{\includegraphics[width=2.4 in]{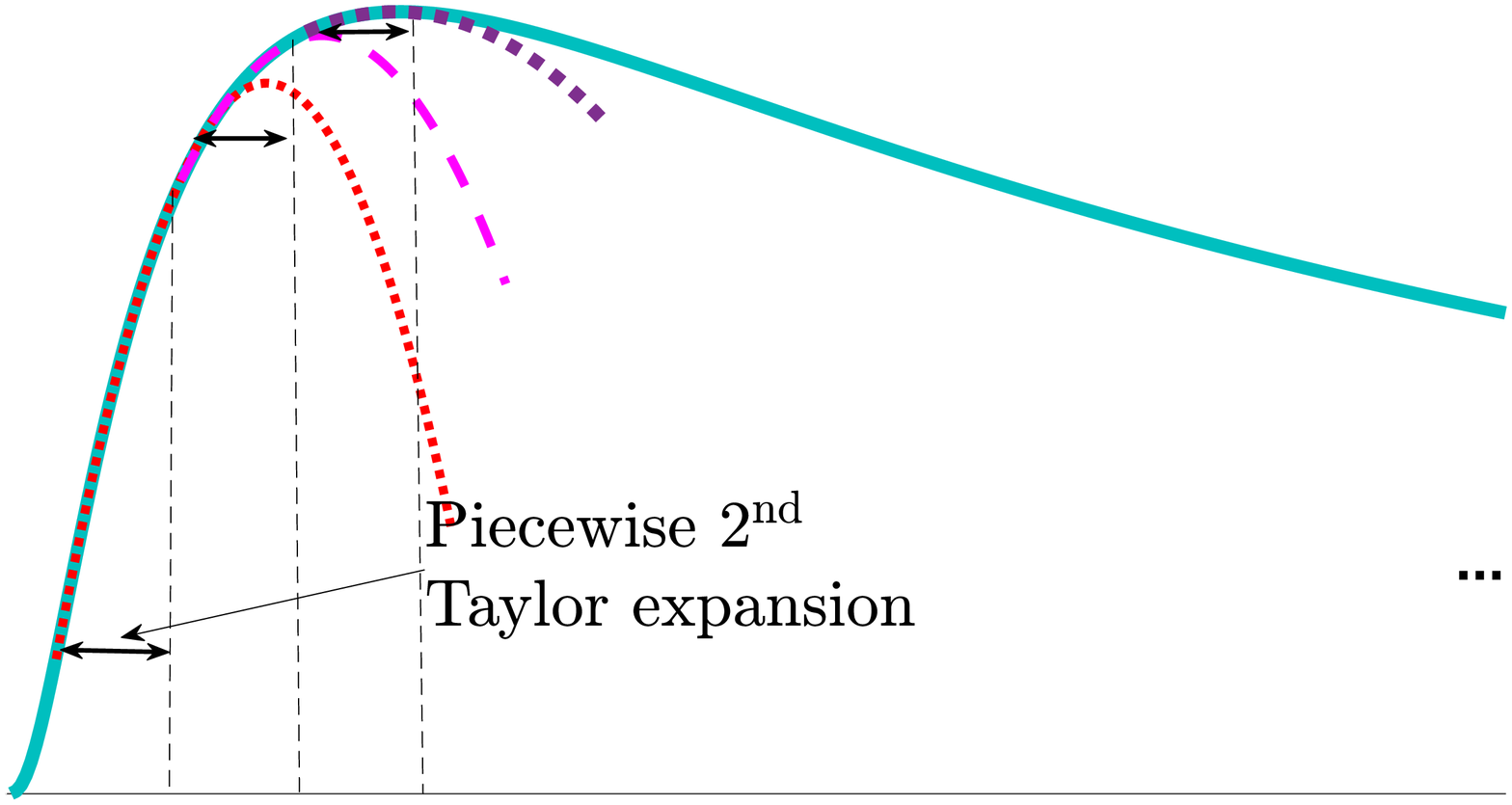}%
\label{fig_Taylor_expansion}}
\caption{Curve of $P_{sm}$ and its $2^{\rm nd}$ Taylor expansion.}
\label{fig:fig_quasi_convex_structure}
\end{figure}

Fact 1 holds because the linear terms w.r.t. $\omega_{sm}$ will donate little to the objective function as the consensus constraint is active. Besides, the quadratic terms pull local values towards the consensus; thus, they will also donate little to the objective function.
Fact 2 holds because the total bandwidth is limited and shared. For example, given a value $\omega_{sm, 2} \in [S_{sm}^{\star}+\delta_{\omega}, W]$ with $\delta_{\omega}$ being a small positive constant, there must exist a value $\omega_{sm,1} \in [\hat \omega_{sm}^{lb}, S_{sm}^{\star}]$ such that $P_{sm}|_{\omega_{sm,1}} = P_{sm}|_{\omega_{sm,2}}$. Thus, a small $\omega_{sm}$ will be preferred as it indicates that more bandwidth can be allocated to URLLC slices to further improve the objective function.

For all $s \in \mathcal{S}^I$, it can be proved that $P_{sm}$ is concave in the interval $(a_1, a_2]$ by evaluating the second-order derivative of $P_{sm}$.
Thus, we can use the $2^{\rm nd}$ Taylor expansion to approximate $P_{sm}$ in this interval.
Considering that $P_{sm}$ is convex in the interval $[\hat \omega_{sm}^{lb}, a_1]$, the $1^{\rm st}$ Taylor expansion is always used to obtain the lower bound of $P_{sm}$.
However, this interval is usually rather narrow, and the value of $P_{sm}$ in this interval is much lower than that in the interval $(a_1, a_2]$. What is more, the error bound of the $1^{\rm st}$ Taylor expansion is greater than that of the $2^{\rm nd}$ expansion. Therefore, we explore the $2^{\rm nd}$ Taylor expansion to approximate $P_{sm}$ in the interval $[\hat \omega_{sm}^{lb}, a_2]$. Fig. \ref{fig_Taylor_expansion} shows an example of the $2^{\rm nd}$ Taylor expansion of $P_{sm}$.
Given a local point $\bm \omega_{m}^{(k,q)}$ at the $q$-th iteration, the Taylor expansion of $-\tilde P_{m}^{(k)}$ at the local point can be given by
\begin{equation}\label{eq:P_sm_Taylor_expansion}
\begin{array}{*{20}{l}}
{ - {{\tilde P}_{m}^{(k)}} \approx  - \tilde P_{m}^{(k,q)} - \nabla \tilde P_m^{(k,q)}{{({{\bm \omega} _m} - {\bm \omega} _m^{(k,q)})}^{\rm T}} - } \\
{\frac{1}{2}({{\bm \omega} _m} - {\bm \omega} _m^{(k,q)})H({\bm \omega} _{m}^{(k,q)}){{({\bm \omega _m} - {\bm \omega} _m^{(k,q)})}^{\rm T}} - R_2(\bm \omega_{m})},
\end{array}
\end{equation}
where ${\bm \omega}_m = [\omega_{1m}, \ldots, \omega_{|{\mathcal S}^I|m}]$, $\nabla \tilde P_m^{(k,q)}$ is the gradient of $\tilde P_m^{(k)}$ over ${\bm \omega_m}$ at the local point ${\bm \omega_m^{(k,q)}}$ with
\begin{equation}\label{eq:first_order_Taylor_expansion}
\begin{array}{l}
\frac{{\partial P_{sm}^{(k)}}}{{\partial \omega _{sm}^{(k,q)}}} = \frac{{{\lambda _s^I}(1 + {\varpi _s}{\rho _o}){e^{ - {\varpi _s}{\sigma ^2}}}}}{{\sum\nolimits_{s \in {{\mathcal S}^I}} {{\lambda _s^I}} }} \\
\qquad \left[ {\frac{{3.5{y_{sm}}{z_s}\omega _{sm}^{2.5(k,q)}}}{{{{({y_{sm}}{z_s} + \omega _{sm}^{(k,q)})}^{4.5}}}} - \frac{{3.5{y_{sm}}\omega _{sm}^{2.5(k,q)}}}{{{{({y_{sm}} + \omega _{sm}^{(k,q)})}^{4.5}}}}} \right],
\end{array}
\end{equation}
and $H({\bm \omega}_m^{(k,q)})$ is a Hessian matrix with
\begin{equation}\label{eq:second_order_Taylor_expansion}
    \begin{array}{*{20}{l}}
\frac{{{\partial ^2}P_{sm}^{(k)}}}{{\partial \omega _{sm}^{2(k,q)}}} = \frac{{{\lambda _s^I}(1 + {\varpi _s}{\rho _o}){e^{ - {\varpi _s}{\sigma ^2}}}}}{{\sum\nolimits_{s \in {{\mathcal S}^I}} {{\lambda _s^I}} }} \times \\
\quad \left[ {\frac{{15.75y_{sm}^2z_s^2\omega _{sm}^{1.5(k,q)}}}{{{{({y_{sm}}{z_s} + \omega _{sm}^{(k,q)})}^{5.5}}}} - \frac{{7{y_{sm}}{z_s}\omega _{sm}^{1.5(k,q)}}}{{{{({y_{sm}}{z_s} + \omega _{sm}^{(k,q)})}^{4.5}}}}} \right] + \\
\quad {\frac{{{\lambda _s^I}(1 + {\varpi _s}{\rho _o}){e^{ - {\varpi _s}{\sigma ^2}}}}}{{\sum\nolimits_{s \in {{\mathcal S}^I}} {{\lambda _s^I}} }}\left[ {\frac{{7{y_{sm}}\omega _{sm}^{1.5(k,q)}}}{{{{({y_{sm}} + \omega _{sm}^{(k,q)})}^{4.5}}}} - \frac{{15.75y_{sm}^2\omega _{sm}^{1.5(k,q)}}}{{{{({y_{sm}} + \omega _{sm}^{(k,q)})}^{5.5}}}}} \right]},
\end{array}
\end{equation}
\begin{equation}\label{eq:second_order_intersection_term}
  \frac{{{\partial ^2}P_{sm}^{(k)}}}{{\partial \omega _{sm}^{(k,q)}\partial \omega _{s'm}^{(k,q)}}} = 0, \forall s \ne s',
\end{equation}
$y_{sm} = {{a{P_{nr,sm}}{P_{ne,sm}}{\lambda_s^I}}}/({{3.5{\lambda _R}}})$, ${z_s} = {{\theta_s^{th}}}/({{ 1 + \theta_s^{th} }})$. Besides, we write $\omega_{sm}^{2.5(k,q)}$ rather than ${( {\omega _{sm}^{(k,q)}} )^{2.5}}$ for lightening the notation.

\begin{lemma}\label{lemma:error_bound}
Let the function $\tilde P_{m}^{(k)}: \mathbb{R}^{|{\mathcal S}^I|} \to \mathbb{R}$ be three times differentiable in a given interval $[\hat \omega_{sm}^{lb}, S_{sm}^{\star}]$ for all $s \in {\mathcal{S}^I}$, then the error bound of $2^{\rm nd}$ degree Taylor expansion of $\tilde P_{m}^{(k)}$ at the local point $\bm \omega_{m}^{(k,q)}$ with $\omega_{sm}^{(k,q)} \in [\hat \omega_{sm}^{lb}, S_{sm}^{\star}]$ is given by
\begin{equation}\label{eq:error_bound_lemma}
\begin{array}{l}
    {R_2}({\bm \omega _m}) = \frac{1}{{3!}}{\left[ {\sum\nolimits_{s \in {{\mathcal S}^I}} {\left( {{\omega _{sm}} - \omega _{sm}^{(k,q)}} \right)\frac{\partial }{{\partial \omega _{sm}^{(k,q)}}}} } \right]^3} \\
    \qquad \max \left\{ {\tilde P_m^{(k)}{|_{\hat {\bm \omega} _m^{lb}}},\tilde P_m^{(k)}{|_{\bm S_m^ \star }}} \right\},
\end{array}
\end{equation}
where $\hat {\bm \omega}_{m}^{lb} = [\hat {\omega}_{1m}^{lb}, \ldots, \hat \omega_{|\mathcal{S}^I|m}^{lb}]$ and $\bm S_{m}^{\star} = [S_{1m}^{\star}, \ldots, S_{|\mathcal{S}^I|m}^{\star}]$.
\end{lemma}
\begin{proof}
Please refer to Appendix C.
\end{proof}

After conducting the $2^{\rm nd}$ Taylor approximation, the objective function becomes a convex function.
Although the constraint (\ref{eq:SAA_admm_problem}b) is $P_{sm}$ $\forall s,m$ related, we need not to conduct the Taylor approximation on (\ref{eq:SAA_admm_problem}b) as $P_{sm}$ is quasi-concave and unimodal. In fact, the probability constraint (\ref{eq:SAA_admm_problem}b) and (\ref{eq:arg_lagarangian}c) are equivalent to the following inequality
\begin{equation}\label{eq:variable_range}
    \hat \omega_{sm}^{lb} \le \omega_{sm} \le \omega_{sm}^{ub},
\end{equation}
where $\omega_{sm}^{ub} \le W$ represents the upper bound of $\omega_{sm}$ satisfying (\ref{eq:SAA_admm_problem}b).

Next, a low-complexity bisection-search-based scheme, the main procedures of which are described below, is developed to obtain $\omega_{sm}^{lb}$, $S_{sm}^{\star}$, and $\omega_{sm}^{ub}$: a) let the function $Q_{sm} = P_{sm} - \pi_s$. Perform the bisection search method \cite{yang2018three} on $Q_{sm} = 0$ to obtain $\omega_{sm}^{lb}$ and $\omega_{sm}^{ub}$ that are the two zero points of $Q_{sm}$; b) with the obtained $\omega_{sm}^{lb}$ and $\omega_{sm}^{ub}$, find the maximum value $S_{sm}^{\star}$ of $P_{sm}$ using the bisection search method again.

According to the above analysis, at the $q$-th iteration, we can rewrite (\ref{eq:bandwidth_beamforming_problem}) as
\begin{subequations}\label{eq:SCA_bandwidth_beamforming_problem}
\begin{alignat}{2}
& \left\{
{\omega _{sm}^{(k + 1,q+1)}},{\bm G_{i,sm}^{(k + 1,q+1)}}
 \right\} = \nonumber \\
& \qquad \mathop {{\rm{argmin}}}\nolimits_{\{\omega _{sm},{\bm G_{i,sm}}\}} \bar {\mathcal L}^{(q)} (\omega _{sm},{\bm G_{i,sm}})  \\
& {\rm s.t. \text{ }}  {\rm for} \text{ } m, (\ref{eq:SAA_admm_problem}c)-(\ref{eq:SAA_admm_problem}e),(\ref{eq:variable_range}) \text{ } \rm{are} \text{ } satisfied,
\end{alignat}
\end{subequations}
where $\bar {\mathcal{L}}^{(q)}(\omega _{sm},\bm G_{i,sm}) =  -{ \frac{1}{M}{\tilde P_{m}^{(k)}} - \frac{{\tilde \rho U_m^{u(k)}}}{M} +}$ ${ \sum\limits_{s \in {{\mathcal S}^I}} {\left[ {\psi _{sm}^{(k,q)}( {\omega _{sm} - \omega _s^{(k,q)}(\bar t)} ) + \frac{\mu }{2}{{\| {\omega _{sm} - \omega _s^{(k,q)}(\bar t)} \|}_2^2}} \right]} }$.

In (\ref{eq:SCA_bandwidth_beamforming_problem}), the objective function is convex, (\ref{eq:SAA_admm_problem}c) is affine, and the constraint (\ref{eq:SAA_admm_problem}d) can be proved to be convex w.r.t. both $\omega_{sm}$ and $\bm G_{i,sm}$ \cite{How2019yang}. Therefore, (\ref{eq:SCA_bandwidth_beamforming_problem}) is a convex problem that can be effectively mitigated by some standard convex optimization tools such as CVX and MOSEK. 

Then, we can summarize the main steps of mitigating the problem (\ref{eq:SAA_admm_problem}) in Algorithm \ref{alg1}.
\begin{algorithm}
\caption{ADMM-based bandwidth allocation algorithm}
\label{alg1}
\begin{algorithmic}[1]
\STATE \textbf{Initialization:} Randomly initialize $\bm G_{i,s}^{(0,0)}$, $\{\omega_{s}^{(0,0)}\}$, let $k_{\rm max}=250$, $q_{\rm max}=250$, $q = 0$, $k = 0$, and generate channel samples \{$\bm H_{i,sm}$\}.
\REPEAT
\REPEAT
\STATE Given $\bm G_{i,sm}^{(k,q)}$, $\omega_{sm}^{(k,q)}$, call the greedy scheme to obtain $b_{i,sm}^{u(k,q+1)}$.
\STATE Optimize (\ref{eq:SCA_bandwidth_beamforming_problem}) with obtained $b_{i,sm}^{u(k,q+1)}$ to achieve $\bm G_{i,sm}^{(k,q+1)}$ and $\omega_{sm}^{(k,q+1)}$. Update $q = q + 1$.
\UNTIL{Convergence or reach at the maximum iteration times $q_{{\rm max}}$.}
\STATE Let $\omega_{sm}^{(k+1,q+1)} = \omega_{sm}^{(k,q+1)}$, update $\psi _{sm}^{(k + 1)}$, $\omega_{s}^{(k+1)}(\bar t)$ using (\ref{eq:psi_update}), (\ref{eq:omega_update}), and update $k = k + 1$.
\UNTIL {Convergence or reach at the maximum iteration times $k_{\max}$.}
\end{algorithmic}
\end{algorithm}

\begin{lemma}\label{lemma_5}
For all $i\in\mathcal{I}_s^u$, $s \in {\mathcal{S}^u}$, and $m \in {\mathcal{M}}$, the obtained power matrix $\bm G_{i,sm}^{(k,q)}$ by Algorithm 1 at the $(k,q)$-th iteration satisfies the low-rank constraint, i.e., the SDR for the power matrix utilized in Algorithm 1 is tight.
\end{lemma}
\begin{proof}
Please refer to Appendix D.
\end{proof}

{
Besides, the computational complexity of Algorithm \ref{alg1} consists of the complexities of calling a greedy scheme, solving an optimization problem with semidefinite matrices and the aggregation of local variables. The complexity of the greedy scheme is $O(N^u)$. There are $N^u$ matrices of size $JK \times JK$ and $|{\mathcal S}^I|$ one-dimensional variables in the optimization problem. An interior-point method is then exploited to solve the optimization problem with the complexity of $O(N^uJ^2K^2 +|{\mathcal S}^I|)^{3.5}$ at the worst-case \cite{ye1997interior}. The complexity of aggregating local variables is $O(M)$. Therefore, the total computational complexity of Algorithm \ref{alg1} is $O(k_{\rm max}(q_{\rm max}(N^u+(N^uJ^2K^2 +|{\mathcal S}^I|)^{3.5})+M))$ at the worst case. Yet, the actual complexity will be much smaller than the worst case.}

\section{Optimization of beamforming and URLLC device association with system sensed channels}
In section V, we obtain a family of global consensus variables $\{\omega_{s}(\bar t)\}$ with the system generated channel samples. The time-varying actual channels may require the re-optimization of beamformers and {URLLC} device associations at each minislot. According to system sensed channels at each minislot, we next discuss how to calculate beamformers and {URLLC} device associations.

At each minislot $t$, given the global consensus variables $\{\omega_s(\bar t)\}$, the original problem (\ref{eq:original_problem}) will be reduced to the following problem
\begin{subequations}\label{eq:mini_time_scale_transformed_problem}
\begin{alignat}{2}
& \mathop {{\rm{maximize}}}\limits_{\{b_{i,s}^u(t), {\bm G}_{i,s}(t)\}} \text{ } {\tilde \rho} {{U}^u(t)}  \\
& {\rm s.t. \text{ }}  \rm {constraints \text{ } (\ref{eq:RRH_energy}), (\ref{eq:total_bandwidth}),(\ref{eq:original_problem}b) \text{ } are \text{ } satisfied.}
\end{alignat}
\end{subequations}

In (\ref{eq:mini_time_scale_transformed_problem}), the channels are system sensed ones at $t$. According to the convexity analysis in section V, (\ref{eq:mini_time_scale_transformed_problem}) is a mixed-integer non-convex programming problem with positive semidefinite matrices, which is hard to be mitigated. Therefore, the alternative optimization scheme presented in subsection V-B can be leveraged to achieve the solutions $b_{i,s}^u(t)$ and $\bm G_{i,s}(t)$ of (\ref{eq:mini_time_scale_transformed_problem}). Lemma \ref{lemma_5} indicates that the achieved ${\rm rank}(\bm G_{i,s}(t)) \le 1$. Thus, we can obtain the beamformers $\bm g_{i,s}(t)$ by performing the eigendecomposition on $\bm G_{i,s}(t)$.
{Summarily,} over a time slot $\bar t$, the slice resource optimization algorithm designed for the RAN slicing system is presented in Algorithm \ref{alg_algo_bandwidth_beamforming}.
\begin{algorithm}
\caption{slice resource optimization algorithm, SRO}
\label{alg_algo_bandwidth_beamforming}
\begin{algorithmic}[1]
\STATE \textbf{Initialization:} $\{{\bm H}_{i,s}(t)\}$, $\forall i \in {{\mathcal I}^u}$, $s \in {{\mathcal S}^u}$, and let $P_{s}^{1} \in [0,1]$, $\vartheta_{a,s}^{1}=0$, $\forall s \in {{\mathcal S}^I}$.
\STATE Call Algorithm \ref{alg1} to obtain $\{\omega_s(\bar t)\}$ for all $s\in {\mathcal S}^I$.
\FOR{$t = 1 : T$}
\STATE Given $\{\omega_s(\bar t)\}$, mitigate (\ref{eq:mini_time_scale_transformed_problem}) by exploiting the alternative optimization scheme to obtain beamformers $\{{\bm g}_{i,s}(t)\}$ and URLLC device associations $b_{i,s}^u(t)$ for all $i \in {\mathcal I}_s^u$, $s \in {\mathcal S}^u$.
\ENDFOR
\end{algorithmic}
\end{algorithm}

\begin{figure}[!t]
\flushleft
\includegraphics[width=3.4in]{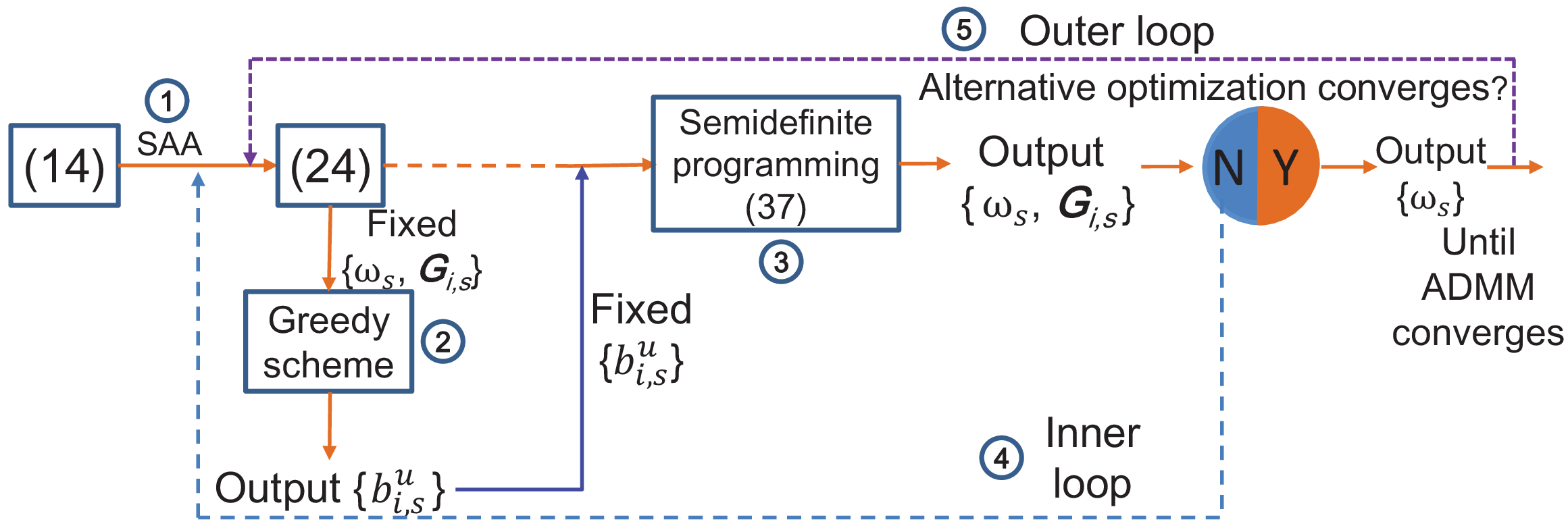}
\caption{{The algorithm logical flow.}}
\label{fig:fig_algorithm_flow}
\end{figure}
{Furthermore, we can depict the logical flow of mitigating (\ref{eq:original_problem}) in Fig. \ref{fig:fig_algorithm_flow}. At the beginning of each time slot $\bar t$, with the system generated channels the RAN-C will follow the flow \textcircled{1} $ \to$ \textcircled{2} $ \to $ \textcircled{3} $\to $ \textcircled{4} $\to$ \textcircled{5} to achieve $\{\omega_s(\bar t)\}$.
The RAN slicing system will allocate $\{\omega_s(\bar t)\}$ bandwidth to mIoT slices, and the remaining system bandwidth is allocated to bursty URLLC slices.
With the achieved $\{\omega_s(\bar t)\}$, the RAN-C acquires sensed channels with which the following flow \textcircled{2} $\to $ \textcircled{3} $\to $ \textcircled{4} will be executed to generate beamformers $\{\bm g_{i,s}(t)\}$ and URLLC device associations $\{b_{i,s}^u (t)\}$.
Next, the RAN slicing system will configure RRHs' transmit beamformers based on $\{\bm g_{i,s}(t)\}$ and establish connections with URLLC devices based on $\{b_{i,s}^u (t)\}$.
}

\section{Simulation results}

\subsection{Comparison algorithms and parameter setting}
{As no existing algorithms can be considered as benchmark algorithms, we design three benchmark algorithms. The effectiveness of the proposed algorithm is verified via comparing it with the benchmark algorithms.} The simulation is also performed to explain the impact of access control schemes on the RAN system performance intuitively. The comparison algorithms are i) {SRO algorithm} that adopts the unrestricted access control scheme; ii) {{SRO-ACB$_{\rm I}$ algorithm}} that utilizes the ACB access control scheme with $P_{ACB} = 0.9$; iii) {SRO-ACB$_{\rm II}$ algorithm} that adopts the ACB access control scheme with $P_{ACB} = 0.5$;
{iv) {Single-sample-based SRO (S$^3$RO) algorithm} that determines the global consensus variable based on a single channel sample instead of the ADMM method.}

The parameter setting is as follows: RRHs and IoT devices are deployed following independent PPPs in a one km$^2$ area. URLLC devices are randomly and uniformly distributed in this area. There are three mIoT slices and two URLLC slices in the RAN slicing system.
For the mIoT slices, set $\vartheta_{w,s}(t) = [1.5, 1.0, 0.5]$, $\pi_s = 0.5$, $\forall s,t$, $\varphi = 4$, $L = 2000$ bits, $\sigma^2 = -90$ dBm, $\rho_o = -90$ dBm, ${\hat E}_{j}^I = 0.03$ mW, $\lambda_R = 3$ RRHs/km$^2$, $\lambda_{s}^I = 18000$ IoT devices/km$^2$, $\forall s$, $a =0.18$ MHz, the queue serving rate $\gamma_s^{th} = a\log_2(1+{\theta_s^{th}})$, $\{\gamma_s^{th}\} = \{5.8, 4.35, 2.9\}$ Kbits/minislot.
For the URLLC slices, the transmit antenna gain at each RRH is set to be $5$ dB, and a log-normal shadowing path-loss model is used to simulate the path-loss between an RRH and a URLLC device with the log-normal shadowing parameter being $10$ dB. A path-loss is computed by $h({\rm dB}) = 128.1 + 37.6\log_{10}d$, where $d$ (in km) is the distance between a device and an RRH. Let $L_{i,s}^u = 160$ bits, $\sigma_{i,s}^2 = -100$ dBm, $\lambda_{s} = \lambda = 0.1$ packets/minislot, $\forall i,s$, $\{I_s^u\} = \{3, 5\}$ devices, and $\{D_s\} = \{1, 2\}$ milliseconds, $E_j = 3$ W, $\forall j$ \cite{tang2019service}.
Other system parameters are shown as follows: $J = 3$, $K = 2$, $\tilde \rho=1$, $\eta =100$, $T = 60$, $W =60$ MHz, $M = 100$, $\mu = 2.9\times10^{-3}$, $\phi =1.5$, $\xi =54$, $\alpha_g=0.05$, $\kappa=5.12 \times 10^{-4}$, $\alpha=10^{-5}$, $\beta=2 \times 10^{-8}$, and $\varsigma =2 \times 10^{-5}$ \cite{How2019yang}.

\subsection{Performance evaluation}
To evaluate the comparison algorithms, the following performance indicators are utilized i) RA success probability $P_s(t)$; ii) expected queue length per IoT device at minislot $t$, $E[Q_s(t)] = \vartheta_{a,s}(t)$; iii) total slice utility $\bar U$ that is the objective function of (\ref{eq:original_problem}).

We first evaluate the convergence mainly determined by that of the ADMM-based framework of the proposed SRO algorithm. We then leverage $\Delta_{\omega} = {\rm{ }}\sum\nolimits_{s \in {{\mathcal S}^I}} {| {\omega _s^{(k + 1)}(\bar t) - \omega _s^{(k)}(\bar t)} |} $ to evaluate the convergence of the SRO algorithm. Fig. \ref{fig:fig_convergence} illustrates the algorithm's convergence. It shows that SRO can converge after several iterations.
\begin{figure}[!t]
\centering
\begin{minipage}[t]{0.45\textwidth}
\centering
\includegraphics[width=2.4in,height=1.2in]{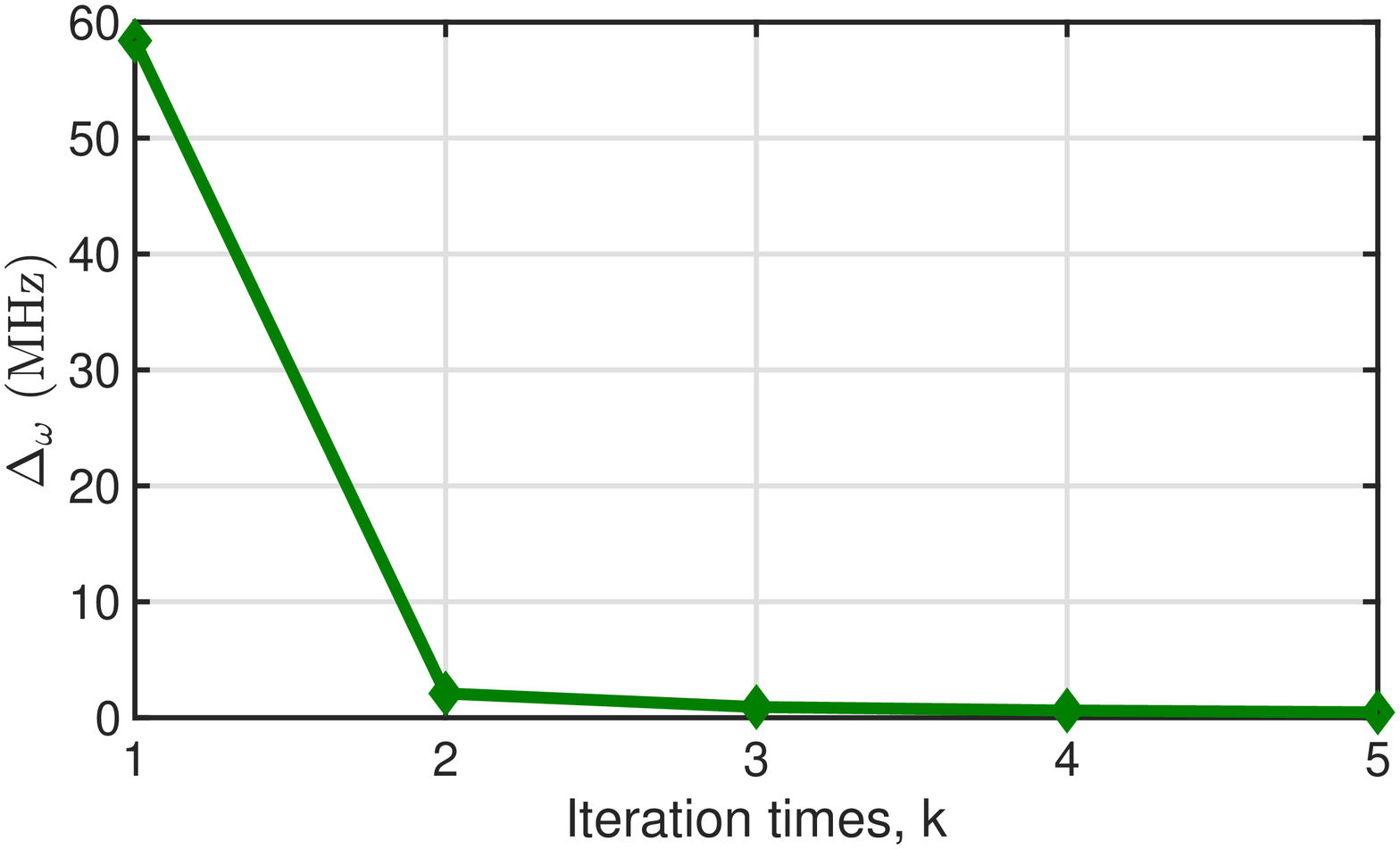}
\caption{The convergence curve of the SRO algorithm.}
\label{fig:fig_convergence}
\end{minipage}
\hspace{0.05\linewidth}
\begin{minipage}[t]{0.45\textwidth}
\centering
\includegraphics[width=2.8in,height=1.6in]{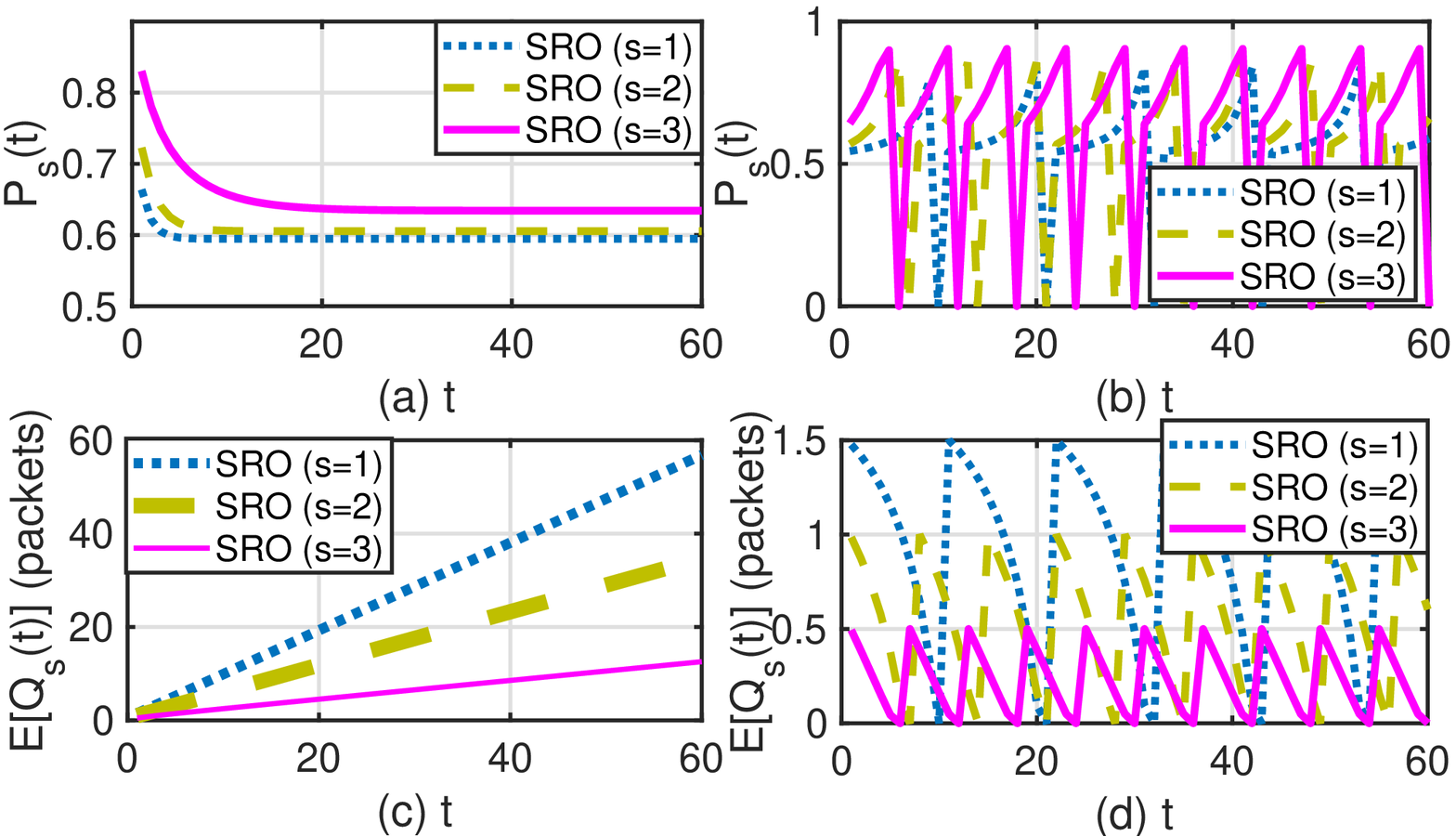}
\caption{Trends of $P_s(t)$ and $E[Q_s(t)]$.}
\label{fig:fig_queueLen_Pst}
\end{minipage}
\end{figure}


We next plot the tendency of the RA success probability $P_s(t)$ and the corresponding expected queue length $E[Q_s(t)]$ during a time slot in Fig. \ref{fig:fig_queueLen_Pst}. Fig. \ref{fig:fig_queueLen_Pst}(a) and \ref{fig:fig_queueLen_Pst}(c) show the tendency of $P_s(t)$ and $E[Q_s(t)]$ in the case of $\{\gamma_s^{th}\} = \{1.8, 1.35, 0.9\}$ Kbits/minislot. Fig. \ref{fig:fig_queueLen_Pst}(b) and \ref{fig:fig_queueLen_Pst}(d) depict the tendency of $P_s(t)$ and $E[Q_s(t)]$ in the case $\{\gamma_s^{th}\} = \{5.8, 4.35, 2.9\}$ Kbits/minislot.

From Fig. \ref{fig:fig_queueLen_Pst}, we obtain the following interesting conclusions: the queue of each IoT device is not stable when the queue serving rate $\gamma_s^{th}$ is small. In this case, the average queue length monotonously increases over $t$. On the contrary, the queue of each IoT device is periodically flushed when a great queue serving rate is configured. The result that the maintained queue by each IoT device can be emptied verifies the correctness of the analysis of the RA process.

\begin{figure}[!t]
\centering
\includegraphics[width=2.7in,height=1.5in]{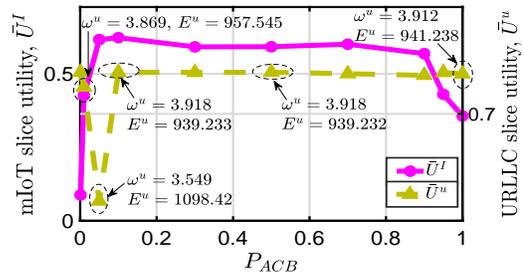}
\caption{{Trends of $\bar U^I$ and $\bar U^u$ vs. $P_{ACB}$.}}
\label{fig:fig_IoT_URLLC_Pacb}
\end{figure}

{As access control schemes have a significant impact on the algorithm performance, we discuss how to select the value of $P_{ACB}$. Fig. \ref{fig:fig_IoT_URLLC_Pacb} depicts the trends of mIoT slice utility $\bar U^I$ and bursty URLLC slice utility $\bar U^u$ w.r.t. $P_{ACB}$ with $P_{ACB} = [0.001, 0.01, 0.05, 0.1, 0.3, 0.5, 0.7, 0.9, 0.95, 1]$, $\lambda_s^I = 19800$ IoT devices/km$^2$, and $\lambda = 1$ packet/minislot.
In this figure, we denote the bandwidth allocated to bursty URLLC slices during a time slot by $\omega^u$ (in MHz) and the sum of RRHs' transmit power for serving URLLC devices during a time slot by $E^u$ (in mW). }

{From this figure, we can observe that: i) the obtained $\bar U^I$ of the proposed algorithm increases with $P_{ACB}$ when $0< P_{ACB} \le 0.1$. This is because more IoT devices have the opportunity to access their corresponding RRHs when the stringent access restriction status is slightly mitigated; ii) when $0.1 < P_{ACB} \le 1$, $\bar U^I$ decreases with $P_{ACB}$. This is because the intra-cell interference is getting worse and more and more IoT devices go into the outage when the access restriction status is further eased; iii) as the system will allocate less bandwidth to bursty URLLC slices and RRHs will consume more transmit power for serving URLLC devices, the obtained $\bar U^u$ is decreased when more IoT devices successfully access the network. However, $\bar U^u$ slightly changes when the value of $P_{ACB}$ becomes greater; iv) the above results indicate that the selection of the value of $P_{ACB}$ should consider the network status (e.g., the interference status).}

Let the IoT device intensity $\bm \lambda^I = [900 n, 900 n, 900 n]$ with $n \in \{6, 8, \ldots, 26\}$.
Under the existence of both mIoT and URLLC slices, we plot trends of the total slice utility $\bar U$ and bursty URLLC slice utility $\bar U^u$ w.r.t. $n$ in Fig. \ref{fig:fig_utility_vs_lamdba_IoT} to understand the impact of the mIoT slices on the performance of all comparison algorithms.
In this figure, with a slight abuse of notation, $B = [b_{11}^u, \ldots, b_{31}^u, b_{12}^u, \ldots, b_{52}^u]$, $\omega^I=[\omega_{SRO}^I, \omega_{ACB_{\rm I}}^I, \omega_{ACB_{\rm II}}^I, \omega_{S^3RO}^I]$ MHz with $\omega_{SRO}^I$, $\omega_{ACB_{\rm I}}^I$, $\omega_{ACB_{\rm II}}^I$, and $\omega_{S^3RO}^I$ representing the bandwidth allocated to mIoT slices by executing SRO, SRO-ACB$_{\rm I}$, SRO-ACB$_{\rm II}$, and S$^3$RO algorithms, respectively. $\bar U^I = [\bar U_{SRO}^I, \bar U_{ACB_{\rm I}}^I, \bar U_{ACB_{\rm II}}^I, \bar U_{S^3RO}^I]$ with $\bar U_{SRO}^I$ denoting the achieved mIoT slice utility of SRO.
\begin{figure}[!t]
\centering
\subfigure[total slice utility vs. $n$.]{\includegraphics[width=2.7in,height=1.5in]{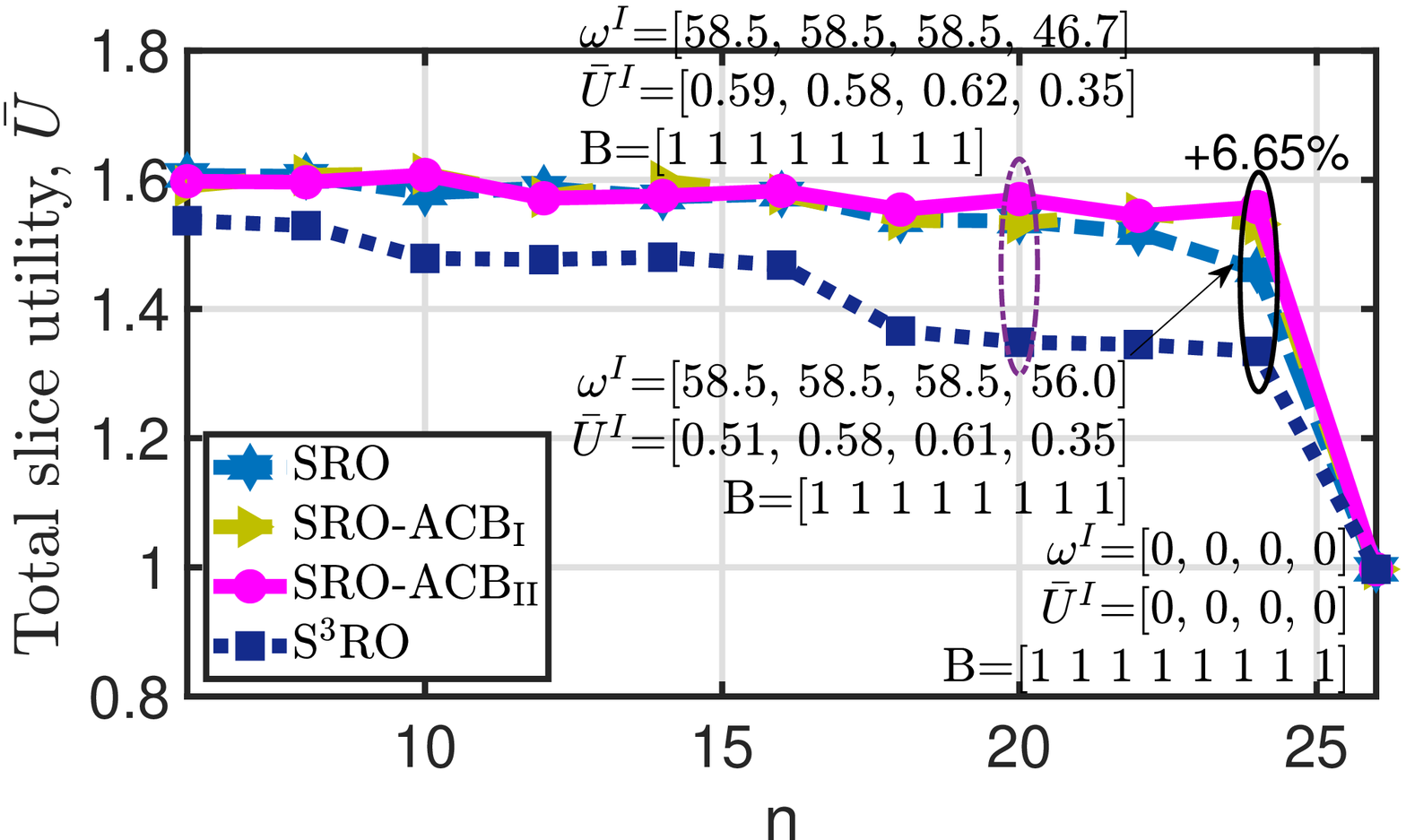}%
\label{fig:fig_total_utility_vs_lamdba_IoT}}
\hspace{0.1\linewidth}
\subfigure[bursty URLLC slice utility vs. $n$.]{\includegraphics[width=2.7in,height=1.5in]{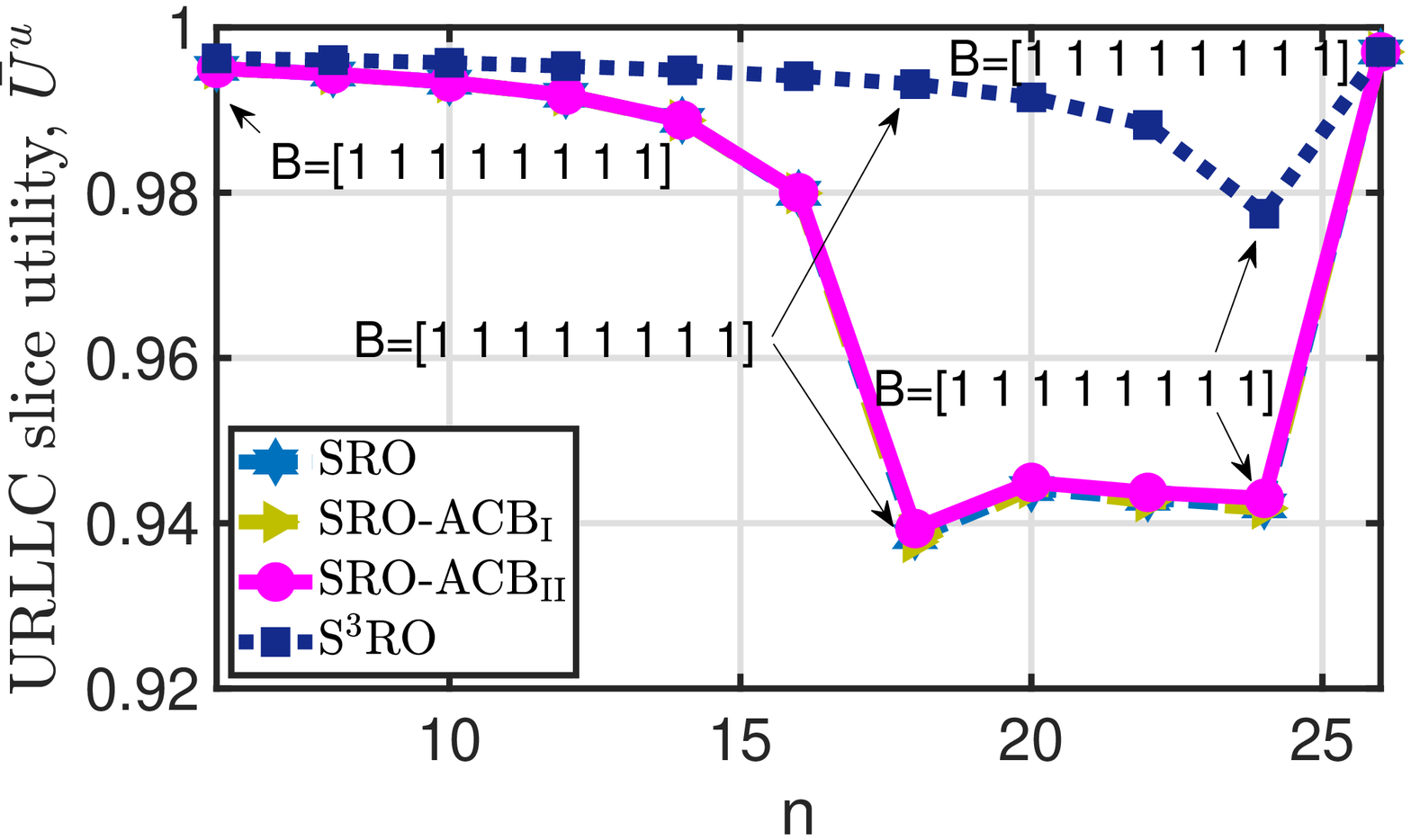}%
\label{fig_vs_Ds}}
\caption{Trends of the achieved total slice utilities and bursty URLLC slice utilities of all algorithms vs. $n$.}
\label{fig:fig_utility_vs_lamdba_IoT}
\end{figure}

The following observations can be obtained from Fig. \ref{fig:fig_utility_vs_lamdba_IoT}:
i) when $n < 16$, all algorithms except for S$^3$RO almost obtain the same $\bar U$, and the obtained utilities are robust to the average number of IoT devices;
ii) when $16 \le n < 26$, the conclusion changes. For the SRO algorithm, its achieved $\bar U$ decreases with an increasing $n$ due to the increasing interference. A great $n$, however, does not cause a significant decrease on the obtained $\bar U$ by SRO-ACB$_{\rm I}$ and SRO-ACB$_{\rm II}$. Thanks to the exploration of an access control scheme, both SRO-ACB$_{\rm I}$ and SRO-ACB$_{\rm II}$ can achieve greater $\bar U$ than SRO. For example, compared with SRO, SRO-ACB$_{\rm II}$ improves $\bar U$ by $6.65\%$ when $n = 24$;
iii) when $n = 26$, which means that the total average number of IoT devices reaches $70,200$ devices, the RAN slicing system fails to create and manage mIoT slices as the QoS requirements of mIoT slices serving such a massive average number of devices cannot be simultaneously satisfied. In this case, all system resources are allocated to URLLC slices, and the maximum bursty URLLC slice utility is obtained;
iv) as mIoT slices and URLLC slices share the system resources, an increasing $n$ results in a decreasing $\bar U^u$;
besides, it is interesting to find that the two access-control-based algorithms may not outperform SRO in terms of obtaining $\bar U^u$. It indicates that URLLC slices do not benefit from access control schemes of mIoT slices when changing $n$;
v) {although the S$^3$RO algorithm can achieve the greatest URLLC slice utility, it obtains the smallest $\bar U$. It indicates that S$^3$RO cannot effectively orchestrate network resources for mIoT and URLLC slices;}
vi) the RAN slicing system can always accommodate the QoS requirements of all URLLC devices.

Next, to understand the impact of URLLC slices on the performance of all comparison algorithms, we plot the trends of $\bar U$ and the mIoT slice utilities obtained by all comparison algorithms w.r.t. URLLC packet arrival rate $\lambda$ with $\lambda = \{0.1,0.5,1.0,\ldots,4.5,5.0\}$ packets per unit time in Fig. \ref{fig:fig_utility_vs_lamdba_URLLC}.
Similarly, the following notations are involved in Fig. \ref{fig:fig_utility_vs_lamdba_URLLC}: $\bar U^u = [\bar U_{SRO}^u, \bar U_{ACB_{\rm I}}^u, \bar U_{ACB_{\rm II}}^u, \bar U_{S^3RO}^u]$ with $\bar U_{SRO}^u$ denoting the URLLC slice utility obtained by running the SRO algorithm.

\begin{figure}[!t]
\centering
\subfigure[total slice utility vs. $\lambda$.]{\includegraphics[width=2.7in,height=1.5in]{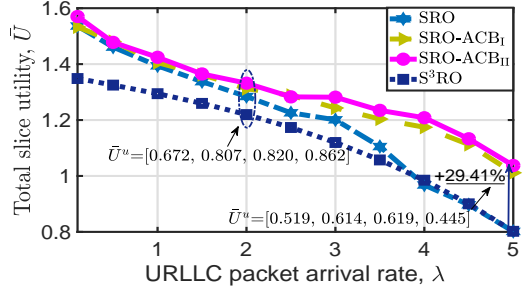}%
\label{fig:fig_total_utility_vs_lamdba_URLLC}}
\hspace{0.1\linewidth}
\subfigure[mIoT slice utility vs. $\lambda$.]{\includegraphics[width=2.7in,height=1.5in]{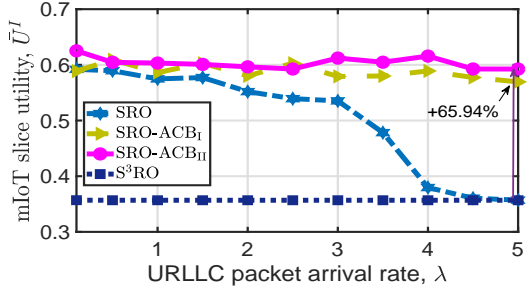}%
\label{fig_vs_Ds}}
\caption{Trends of the achieved total slice utilities and IoT slice utilities of all algorithms vs. $\lambda$.}
\label{fig:fig_utility_vs_lamdba_URLLC}
\end{figure}

From Fig. \ref{fig:fig_utility_vs_lamdba_URLLC}, we can observe that:
i) the obtained utilities $\bar U$ of all algorithms decrease with $\lambda$ mainly due to the decrease of the bursty URLLC slice utility. Two algorithms adopting the access control scheme always achieve greater utilities $\bar U$ than SRO. For example, when $\lambda = 5$, compared with the SRO algorithm, the obtained $\bar U$ of SRO-ACB$_{\rm II}$ is increased by $29.41\%$;
ii) SRO-ACB$_{\rm II}$ may achieve greater $\bar U$ than SRO-ACB$_{\rm I}$ as a greater $\bar U^I$ is obtained by reducing more interfering IoT devices;
iii) the obtained mIoT slice utilities $\bar U^I$ of SRO-ACB$_{\rm I}$, SRO-ACB$_{\rm II}$, and S$^3$RO are robust to the URLLC packet arrival rate. The obtained $\bar U^I$ of SRO decreases with an increasing $\lambda$;
iv) an important observation is that the $\bar U^I$ of the access-control-based SRO-ACB$_{\rm I}$ algorithm is $1.65$ times that of the SRO algorithm when $\lambda = 5$.
It explicitly reflects that mIoT slices can still benefit from access control schemes even though $\lambda$ is changed.

Figs. \ref{fig:fig_utility_vs_lamdba_IoT} and \ref{fig:fig_utility_vs_lamdba_URLLC} illustrate the situation of a given total system bandwidth. We next change the total bandwidth $W$ and plot its impact on the obtained $\bar U$ of all algorithms in Fig. \ref{fig:fig_total_utility_vs_bandwidth}. {The following notations are used in this figure: $\omega^u = [\omega_{SRO}^u, \omega_{ACB_{\rm I}}^u, \omega_{ACB_{\rm II}}^u, \omega_{S^3RO}^u]$ MHz with $\omega_{SRO}^u$ denoting the bandwidth allocated to URLLC slices by running the SRO algorithm.}
\begin{figure}[!t]
\centering
\begin{minipage}[t]{0.45\textwidth}
\centering
\includegraphics[width=2.8in,height=1.6in]{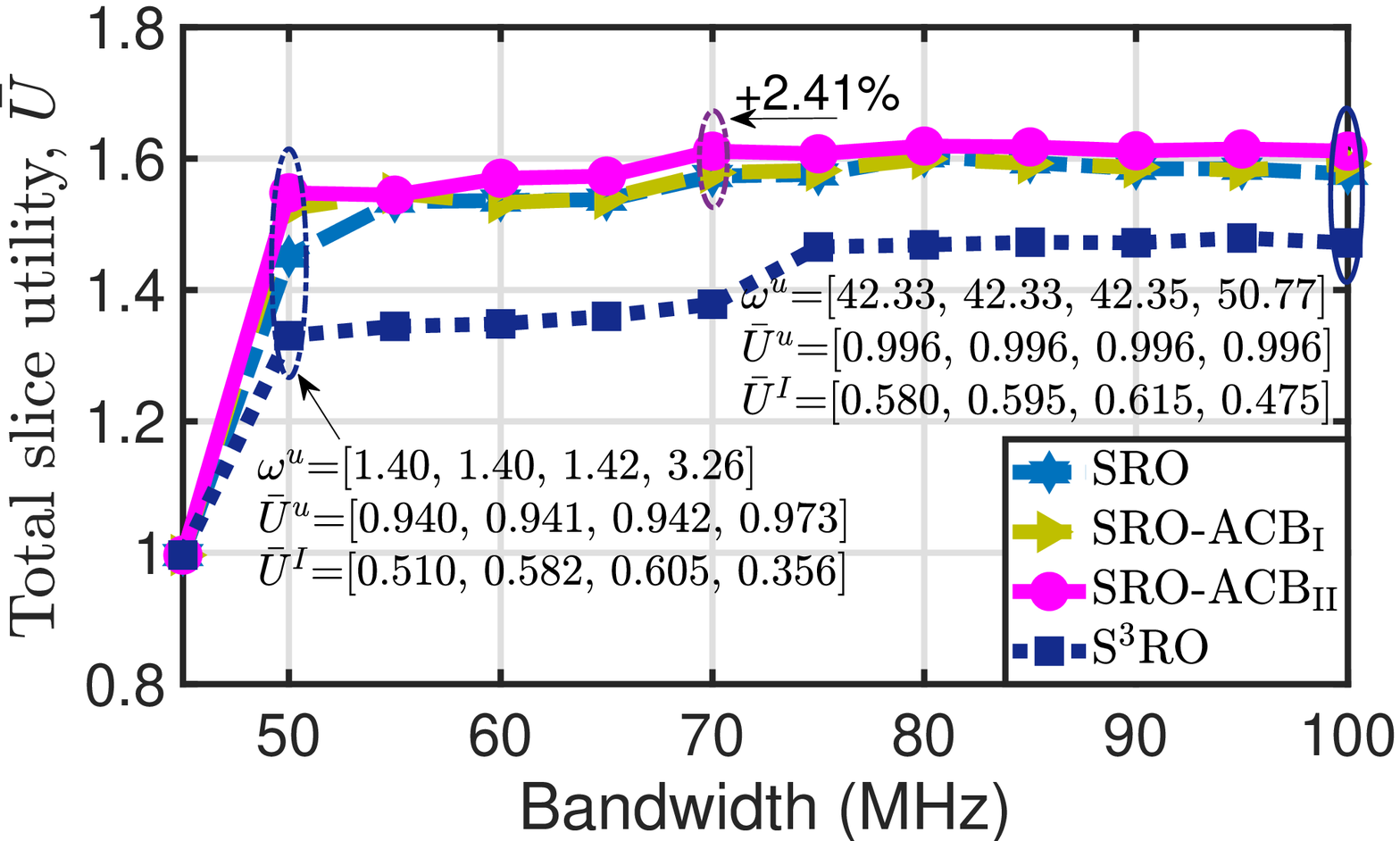}
\caption{Trend of achieved $\bar U$ vs. system bandwidth.}
\label{fig:fig_total_utility_vs_bandwidth}
\end{minipage}
\hspace{0.05\linewidth}
\begin{minipage}[t]{0.45\textwidth}
\centering
\includegraphics[width=2.8in,height=1.6in]{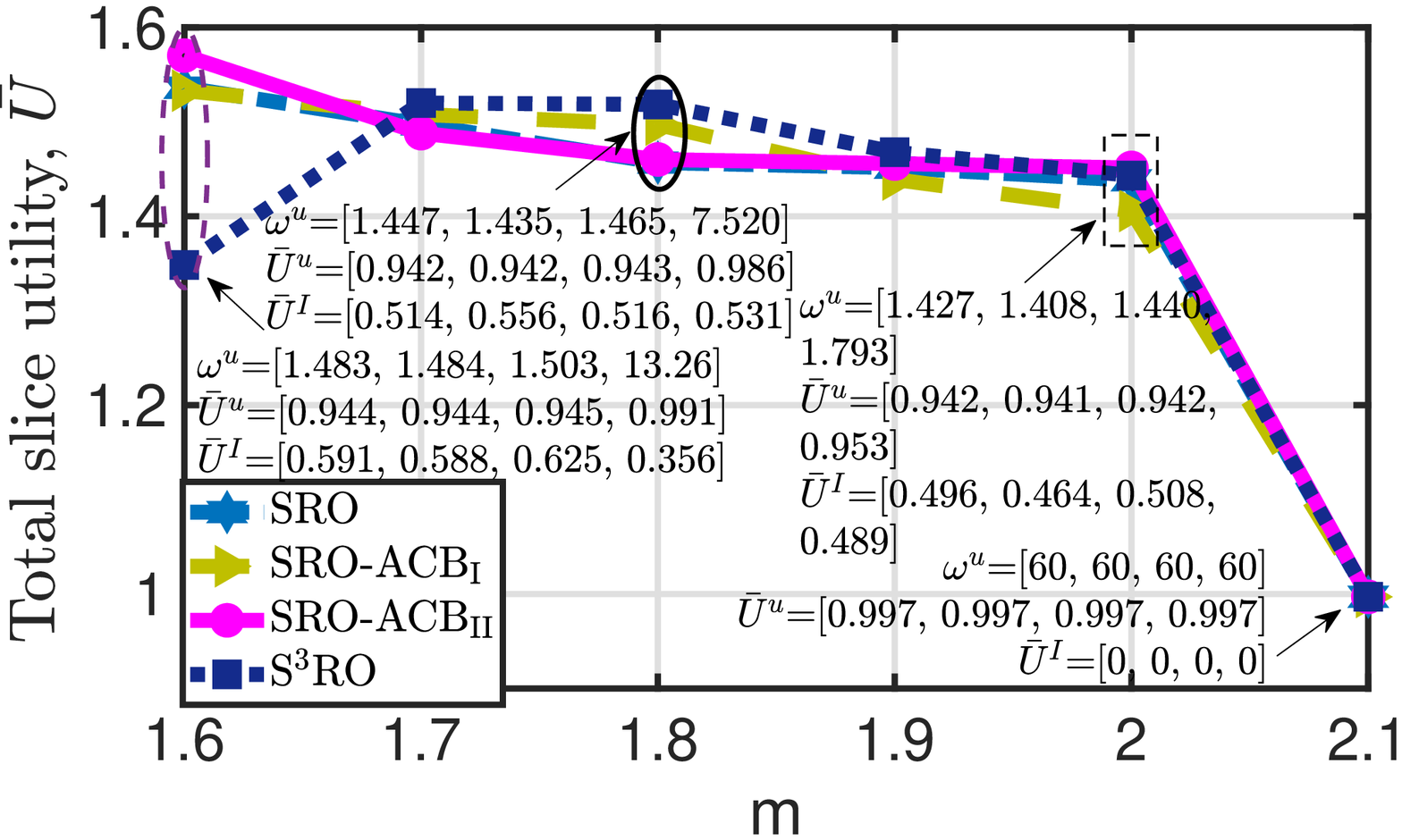}
\caption{Trend of achieved total slice utility vs. $m$.}
\label{fig:fig_total_utility_vs_gammath}
\end{minipage}
\end{figure}


The following conclusions can be obtained from Fig. \ref{fig:fig_total_utility_vs_bandwidth}:
i) when $W=45$ MHz, the QoS requirements of all IoT devices cannot be simultaneously satisfied. As a result, the total bandwidth is allocated to URLLC slices;
ii) when $W$ locates in the range of $(45, 55]$ MHz, the achieved total slice utilities $\bar U$ of SRO and SRO-ACB$_{\rm I}$ increase with $W$. Owing to the utilization of the access control scheme, SRO-ACB$_{\rm I}$ and SRO-ACB$_{\rm II}$ obtain higher $\bar U$ than SRO. For example, compared with the SRO algorithm, the SRO-ACB$_{\rm II}$ algorithm improves the achieved $\bar U$ by $6.66\%$ when $W = 50$ MHz;
iii) when $W > 55$ MHz, all algorithms cannot remarkably improve $\bar U$;
{iv) S$^3$RO achieves the smallest $\bar U$ under different bandwidth values.}

We also discuss other crucial parameters' impact on the performance of the comparison algorithms.
We reconfigure $\{\gamma_s^{th}\}$ of mIoT slices as $\gamma_1^{th} = 3.6 m$, $\gamma_2^{th} = 2.7 m$ and $\gamma_3^{th} = 1.8 m$ Kbits/minislot with $m \in \{1.5, 1.6, \ldots, 2.1\}$ and $\{D_s\}$ of URLLC slices as $D_1 = 0.00025 d$ second and $D_2 = 0.0005 d$ second with $d \in \{2, 3, \ldots, 10\}$. The impact of QoS requirements of network slices on the total slice utility is plotted in Figs. \ref{fig:fig_total_utility_vs_gammath} and \ref{fig:fig_total_utility_vs_Ds}. The impact of energy efficiency coefficient $\eta$ is plotted in Fig. \ref{fig:fig_total_utility_vs_eta}. In these figures, we denote the power consumption of RRHs of all algorithms by $E^u = [E_{SRO}^u, E_{ACB_{\rm I}}^u, E_{ACB_{\rm II}}^u, E_{S^3RO}^u]$ mW with $E_{SRO}^u = \sum\limits_{t = 1}^T {\sum\limits_{s \in {{\mathcal S}^u}} {\sum\limits_{i \in {\mathcal I}_s^u} {b_{i,s}^u{\rm tr}({\bm G_{i,s}})} } } $.


\begin{figure}[!t]
\centering
\begin{minipage}[t]{0.45\textwidth}
\centering
\includegraphics[width=2.8in,height=1.6in]{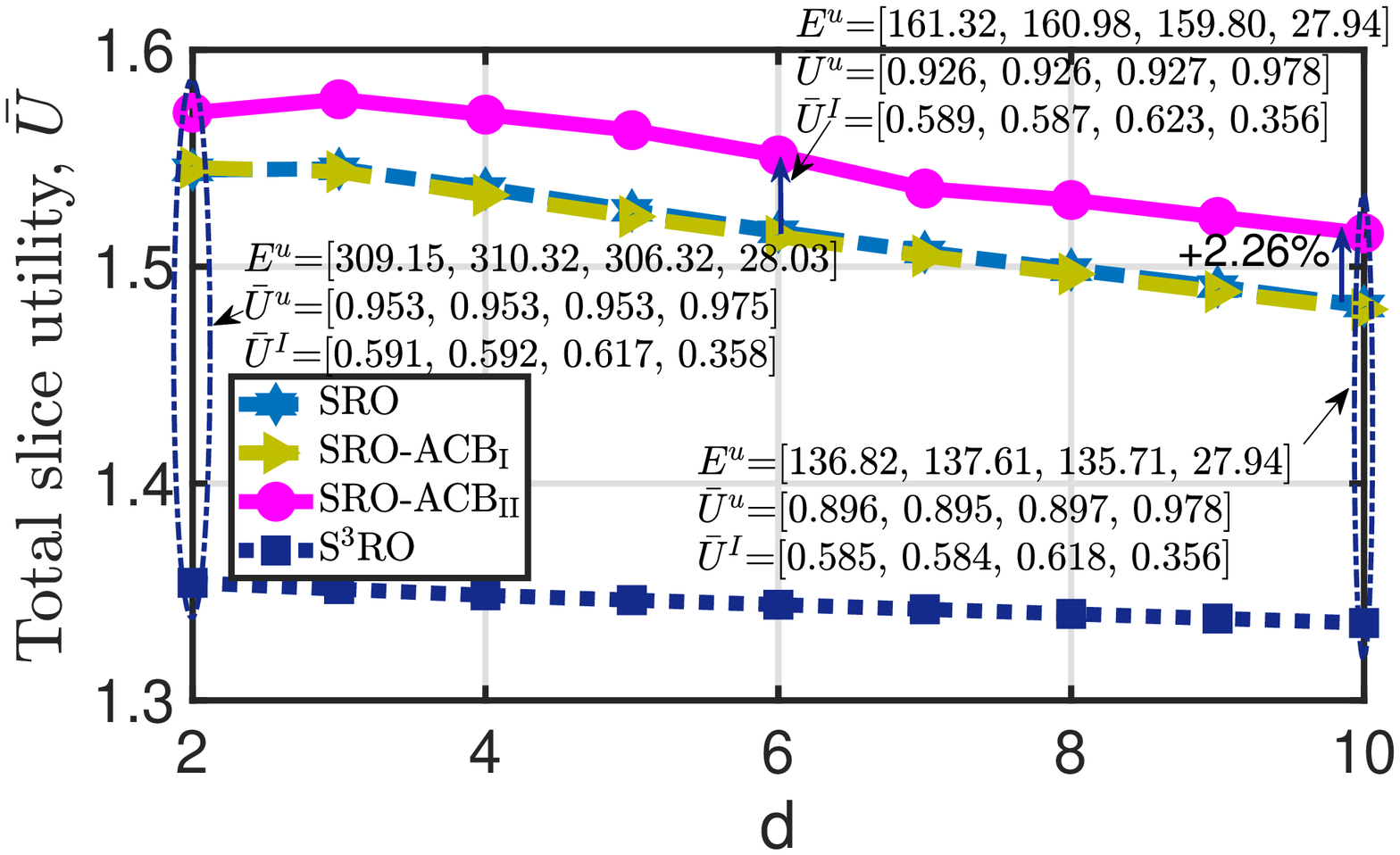}
\caption{Trend of achieved total slice utility vs. $d$.}
\label{fig:fig_total_utility_vs_Ds}
\end{minipage}
\hspace{0.05\linewidth}
\begin{minipage}[t]{0.45\textwidth}
\centering
\includegraphics[width=2.8in,height=1.6in]{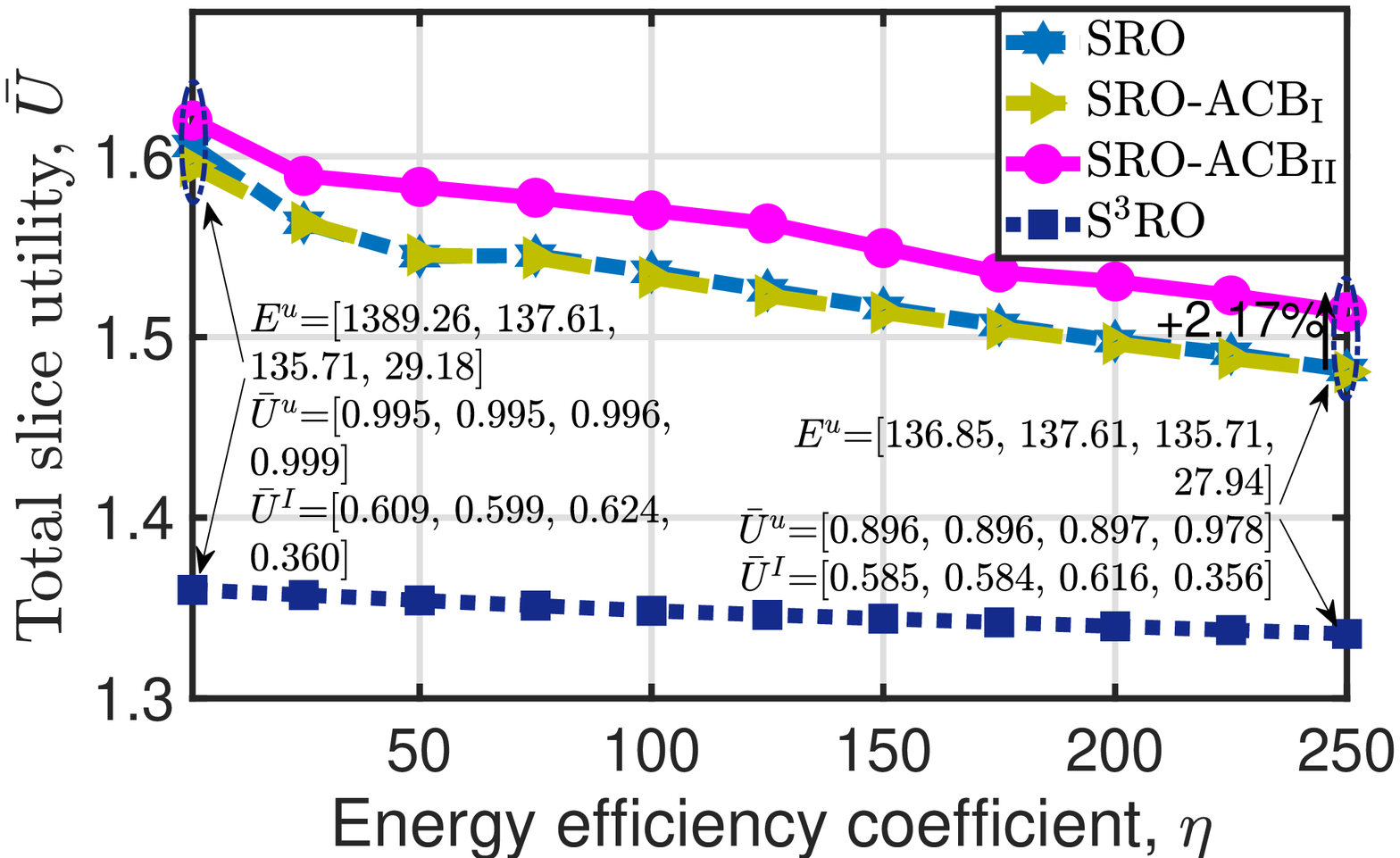}
\caption{Trend of achieved total slice utility vs. $\eta$.}
\label{fig:fig_total_utility_vs_eta}
\end{minipage}
\end{figure}


From these figures, the following observations can be achieved:
i) the obtained utilities $\bar U$ of all algorithms apart from S$^3$RO decrease with an increasing $m$. This is because a great $m$ indicates that the accumulated IoT packets in the queue of each IoT device can be quickly emptied, and then a small $P_s(t)$ is obtained. {For S$^3$RO, it achieves a fluctuating $\bar U$ as only one channel sample is used to orchestrate slice resources;}
ii) a great $D_s$ will reduce RRHs' power consumption. However, it also reduces the URLLC slice gain. Then, it may be hard to conclude the trend of $\bar U^u$ w.r.t. $D_s$ as the energy efficiency coefficient $\eta$ significantly affects the value of $\bar U^u$;
iii) it is also uneasy to conclude the trend of $\bar U^u$ w.r.t. $\eta$. An increasing $\eta$ causes a decrease of RRHs' power consumption. Yet, the value of $\bar U^u$ is determined by the multiplier of $\eta$ and $E^u$;
iv) the SRO-ACB$_{\rm II}$ algorithm may perform better than the SRO algorithm. However, the performance of the other access-control-based algorithm, SRO-ACB$_{\rm I}$, is slightly worse than SRO.
Besides, it cannot ensure that the $\bar U^I$ obtained by the access-control-based algorithms are always higher than that of SRO. At sometimes, access control schemes may drag down the utility of the mIoT service;
{v) S$^3$RO obtains the smallest $\bar U$, which further indicates that S$^3$RO cannot effectively address the two-timescale issue.}


\section{Conclusion}
In this paper, we extended the subframe structure of a RAN slicing system to admit more IoT devices and analyzed the RACH of a randomly chosen IoT device. Based on the analysis result, we derived closed-form expressions of RA success probabilities of devices with an unrestricted access control scheme and an ACB access control scheme.
Next, we formulated the RAN slicing for mIoT and bursty URLLC service multiplexing as an optimization problem to optimally orchestrate RAN resources for mIoT slices and bursty URLLC slices. Efficient mechanisms such as SAA and ADMM were then exploited to mitigate the optimization problem.
Simulation results showed that the proposed algorithm could support more IoT devices and could effectively implement the service multiplexing of mIoT and bursty URLLC traffic.
This paper focused on the {orthogonal} RAN slicing system for mIoT and bursty URLLC service multiplexing provision. The {non-orthogonal} RAN slicing system is a topic worthy of research in the near future.

\appendix

\subsection{Proof of Lemma 1}
The work in \cite{jiang2018random} adopted a standard stochastic geometry method to derive the LT of the aggregate interference from interfering IoT devices. Different from \cite{jiang2018random}, both the stochastic geometry method and a gamma-Poisson distribution are exploited to derive the result in this paper.

For the origin RRH, the LT of its interference from interfering IoT devices in $s$ can be derived as
\begin{equation}\label{eq:LT_intra_interference}
    \begin{array}{l}
{\mathcal L}_{{{\mathcal I}_{s}(t)}} (\varpi_s) =  {{\mathbb E}_{{\mathcal I}_{s}(t)}}\left[ {{e^{ - \varpi_s {\mathcal I}_{s}(t)}}} \right] \\
 = {{\mathbb E}_{{\mathcal I}_{s}(t)}}[ {\rm exp}\{ - \varpi_s \sum\limits_{m \in {{ u}_{s}^I} \backslash \{ o\} } {\mathbbm  1}({p_m}||{d_m}|{|^{ - \varphi }} = {\rho _o}) \\
\quad {\mathbbm  1}(N_{a,s}{(t)} > 0){\mathbbm  1}({f_m} = {f_o}){\rho _o}{h_m} \} ] \\
\mathop  = \limits^{(a)} {{\mathbb E}_{{{ u}_{s}^I}} }[ \prod\limits_{m \in {{ u}_{s}^I} \backslash \{ o\} } {{\mathbb E}_{{h_m}}}[ {\rm exp}\{ - \varpi_s {\mathbbm  1}({P_m}||{u_{m,s}}|{|^{ - \varphi }} = {\rho _o}) \\
\quad {\mathbbm  1}(N_{a,s}{(t)} > 0){\mathbbm  1}({f_m} = {f_o}){\rho _o}{h_m} \}  ]  ]\\
\mathop  = \limits^{(b)} \sum\limits_{n = 0}^\infty  {{\mathbb P}\{ |Z_s| = n\} \prod\limits_{m \in Z_s} {{{\mathbb E}_{{h_m}}}\left[ {{e^{ - \varpi_s{\rho _o}{h_m}}}} \right]} }  \\
\mathop  = \limits^{(c)} {\mathbb P}\{ |Z_s| = 0\}  + \sum\limits_{n = 1}^\infty  {{\mathbb P}\{ |Z_s| = n\} {{\left( {\frac{1}{{1 + \varpi_s{\rho _o}}}} \right)}^n}} \\
\mathop  = \limits^{(d)} {\tilde {\mathbb P}}_{X_s}\{ {X_s} = 1\}  + \left\{ {\sum\limits_{n' = 0}^\infty  {{{\tilde {\mathbb P}}_{X_s}}\{ {X_s} = n'\} {{\left( {\frac{1}{{1 + \varpi_s{\rho _o}}}} \right)}^{n'}}}  - } \right. \\
\left. {\sum\limits_{n' = 0}^1 {{{\tilde {\mathbb P}}_{X_s}}\{ {X_s} = n'\} {{\left( {\frac{1}{{1 + \varpi_s{\rho _o}}}} \right)}^{n'}}} } \right\}(1 + \varpi_s{\rho _o}),
\end{array}
    \end{equation}
where $\varpi_s = \frac{\theta_s^{th}}{\rho_o}$, $Z_s$ denotes the set of interfering IoT devices associated with the origin RRH in mIoT slice $s$, {$|Z_s|$ is the number of devices in $Z_s$,} $X_s$ represents the number of active IoT devices associated with the origin RRH in $s$. According to the conclusion of Lemma 1 in \cite{yu2013downlink}, the probability mass function (PMF) ${\tilde {\mathbb P}}_{X_s} \{ X_s = n'\}$ can be written as
\begin{equation}\label{eq:intra_interference_proba}
    {\tilde {\mathbb P}}_{X_s}\{ X_s = n'\}  = \frac{{{{3.5}^{3.5}}\Gamma (n' + 3.5){{(\frac{{{P_{nr,s}}(t)P_{ne,s}(t){\lambda _{s}^I}}}{{{\lambda _R}\xi F_s}})}^{n' }}}}{{\Gamma (3.5)(n')!{{(\frac{{P_{nr,s}}(t){P_{ne,s}(t){\lambda _{s}^I}}}{{{\lambda _R} \xi F_s}} + 3.5)}^{n' + 3.5}}}},
\end{equation}
with $\Gamma(\cdot)$ being the gamma function.
Besides, in (\ref{eq:LT_intra_interference}), (a) follows from the i.i.d distribution of $h_m$ and its further independence from the Poisson point process $\Phi_s$ or $u_s^I$; (b) follows from the expectation of a discrete random variable; (c) follows from the LT over $h_m$; (d) follows from the fact that the number of active IoT device in a cell is one more than the number of active interfering IoT devices in this cell.

From (\ref{eq:intra_interference_proba}), we can deduce that $X_s$ ($s\in \mathcal{S}^I$) is a gamma-Poisson random variable with $X_s \sim {\rm gamma}$ ${\rm -Poisson}(\alpha_s, 3.5)$ and $\alpha_s = \frac{{{P_{nr,s}}(t){P_{ne,s}}(t){\lambda _{s}^I}}}{{3.5{\lambda _R}\xi F_s}}$.
For a gamma-Poisson random variable $X_s \sim {\rm gamma}$ ${\rm -Poisson}(\alpha, \beta)$, the following expression holds: ${\mathbb E}[e^{X_s}] = (1+\alpha-\alpha e)^{-\beta}$. Thus, we can rewrite  (\ref{eq:LT_intra_interference}) as (\ref{eq:LT_interference_expression}).
This completes the proof.

\subsection{Proof of Lemma 2}

As new endogenous packet arrivals in any IoT device at each minislot $t$ is modelled as a Poisson distribution, the departure process of packets can be regarded as an approximated thinning process of new arrivals, where the thinning factor is related to the RA success probability. The number of accumulated packets in the queue of any IoT device can then be approximated as a Poisson distribution with intensity $\mu_{a,s}^{t}$ ($s\in \mathcal
S^I$) after the thinning process in a specific minislot $t$ ($t > 1$) \cite{jiang2018random}.

Thus, we can derive the expression of $\mu_{a,s}^{t}$ ($t > 1$) via combining with the following facts
\begin{itemize}
    \item {\textbf{Fact 1:}} the accumulated packets during the $t-1$-th minislot will contribute to the accumulated packets at the $t$-th minislot.
    \item {\textbf{Fact 2:}} the arrival packets during the $t-1$-th minislot will also contribute to the accumulated packets in the queue of an IoT device at the $t$-th minislot.
    \item {\textbf{Fact 3:}} an IoT device can send packets only if its preamble is successfully transmitted.
    \item {\textbf{Fact 4:}} at the same minislot, the new packet arrival process and the packet accumulated process are independent.
\end{itemize}

Similar as the Theorem 2 in \cite{jiang2018random}, we can infer that at the $2^{\rm nd}$ minislot, for all $s \in \mathcal{S}^I$, $\mu_{a,s}^{2}$ depends on the intensity of new packet arrivals $\mu_{w,s}^{1}$ and the probability $P_{s}^{1}$ of a randomly selected IoT device at the $1^{\rm st}$ minislot, which is given by
\begin{equation}\label{eq:mu_acc_2_1}
    \begin{array}{l}
\mu _{a,s}^2 = \mu _{w,s}^1 - x_s P_s^1\left( {1 - {e^{ - \mu _{w,s}^1}}} \right).
\end{array}
\end{equation}

The detailed proof of (\ref{eq:mu_acc_2_1}) is omitted for brevity, and a similar proof can be found in Theorem 2 in \cite{jiang2018random}.

Considering that $\mu_{a,s}(t)$ is non-negative at each minislot $t$, we have
\begin{equation}\label{eq:mu_acc_2}
    \begin{array}{*{20}{l}}
{\mu _{a,s}^2 = }
\end{array}{\left[ {\mu _{w,s}^1 - {x_s}P_s^1\left( {1 - {e^{ - \mu _{w,s}^1}}} \right)} \right]^ + }.
\end{equation}

Then, according to the definition of non-empty probability and the Poisson approximation, the non-empty probability of a randomly selected IoT device in mIoT slice $s \in \mathcal{S}^I$ at the $2^{\rm nd}$ minislot can be approximated as 
\begin{equation}\label{eq:non_empty_prob_2_expression}
    P_{ne,s}^2 = 1 - {\mathbb P}\{N_{a,s}^2 = 0\} = 1 - {e^{ - \mu _{a,s}^2}}.
\end{equation}

At the $3^{\rm rd}$ minislot, the intensity of accumulated data packets in the queue of a randomly selected IoT device can be derived as the following
\begin{equation}\label{eq:mu_acc_3}
    \begin{array}{l}
\mu _{a,s}^3 = P_s^2\left( {\sum\limits_{n = 1}^\infty  {\left( {{{[n - x_s]}^ + } \sum\limits_{z = 0}^n {{P_{N_{w,s}^2}}(z){P_{N_{a,s}^2}}(n - z)} } \right)} } \right)
 + \\
\qquad (1 - P_s^2)\left( {\sum\limits_{n = 1}^\infty  {n  \sum\limits_{z = 0}^n {{P_{N_{w,s}^2}}(z){P_{N_{a,s}^2}}(n - z)} } } \right)\\
\mathop  = \limits^{(a)} P_s^2\left [ {\sum\limits_{n = 1}^\infty  {\sum\limits_{z = 0}^n {\frac{{{{\left( {\mu _{w,s}^2} \right)}^z}{e^{ - \mu _{w,s}^2}}}}{{z!}}\frac{{{{\left( {\mu _{a,s}^2} \right)}^{n - z}}{e^{ - \mu _{a,s}^2}}}}{{(n - z)!}}}  \times n}  - } \right. \\
\qquad \quad  \left. {x_s\sum\limits_{n = 1}^\infty  {\sum\limits_{z = 0}^n {\frac{{{{\left( {\mu _{w,s}^2} \right)}^z}{e^{ - \mu _{w,s}^2}}}}{{z!}}\frac{{{{\left( {\mu _{a,s}^2} \right)}^{n - z}}{e^{ - \mu _{a,s}^2}}}}{{(n - z)!}}} } } \right ]^+ + \\
\qquad \quad (1 - P_s^2)\sum\limits_{n = 1}^\infty  {\sum\limits_{z = 0}^n {\frac{{{{\left( {\mu _{w,s}^2} \right)}^z}{e^{ - \mu _{w,s}^2}}}}{{z!}}\frac{{{{\left( {\mu _{a,s}^2} \right)}^{n - z}}{e^{ - \mu _{a,s}^2}}}}{{(n - z)!}}}  \times n} \\
\mathop  = \limits^{(b)} \left [ \mu _{w,s}^2 + \mu _{a,s}^2 - x_sP_s^2\left( {1 - {e^{ - \mu _{w,s}^2 - \mu _{a,s}^2}}} \right) \right ]^+,
\end{array}
\end{equation}
where $P_{N_{w,s}^2}$ and $P_{N_{a,s}^2}$ represent the PMFs of new arrival packets and accumulated packets at the $2^{\rm nd}$ minislot, respectively. Besides, (a) follows from the fact: for any two independent Poisson distributions $\Phi_{X_1}$ and $\Phi_{X_2}$, ${\mathbb P}_{X_1,X_2}({X_1} + {X_2} = x) = \sum\limits_{y = 0}^x {{\mathbb P}_{X_1}({X_1} = y){\mathbb P}_{X_2}(X_2 = x - y)} $; (b) holds as $\Phi_{X_1,X_2}$ is a two dimensional Poisson distribution with an intensity $\lambda_{X_1} + \lambda_{X_2}$, and $\sum\limits_{x = 1}^\infty  {{\mathbb P}_{X_1, X_2}({X_1} + {X_2} = x)}  = 1 - {\mathbb P}_{X_1,X_2}({X_1} + {X_2} = 0)$.

Similarly, we have 
\begin{equation}\label{eq:non_empty_prob_3}
    P_{ne,s}^3 = 1 - {\mathbb P}\{N_{a,s}^3 = 0\} = 1 - {e^{ - \mu _{a,s}^3}}.
\end{equation}

When $t > 3$, since the accumulated packets evolution model of the queue of any IoT device is the similar as that at $t = 3$, we can extend the conclusion obtained at $t = 3$ to that at $t > 3$.

Therefore, we can obtain the closed-form expression of $\mu_{a,s}^t$ for all $s \in \mathcal{S}^I$ at $t > 1$ with
\begin{equation}\label{eq:mu_acc_m_expression_proof}
    \mu _{a,s}^t = \left [ \mu _{w,s}^{t - 1} + \mu _{a,s}^{t - 1} - x_sP_s^{t - 1}\left( {1 - {e^{ - \mu _{w,s}^{t - 1} - \mu _{a,s}^{t - 1}}}} \right) \right ]^+,
\end{equation}
and 
\begin{equation}\label{eq:non_empty_prob_m_expression_proof}
    P_{ne,s}^t = 1 - {\mathbb P}\{N_{a,s}^t = 0\} = 1 - {e^{ - \mu _{a,s}^t}}.
\end{equation}

This completes the proof.

\subsection{Proof of Lemma 3}
The $2^{\rm nd}$ degree Taylor expansion of $\tilde P_{m}^{(k)}$ at the local point $\bm \omega_{m}^{(k,q)}$ is
\begin{equation}
\tilde P_{2,m}^{(k)} = \sum\limits_{j = 0}^2 {\frac{1}{{j!}}{{\left[ {\sum\limits_{s \in {\mathcal S}^I} {\left( {{\omega _{sm}} - \omega _{sm}^{(k,q)}} \right)\frac{\partial }{{\partial \omega _{sm}^{(k,q)}}}} } \right]}^j}{\tilde P_m^{(k)}}|_{\bm \omega_m^{(k,q)}}}.
\end{equation}

The $3^{\rm rd}$ degree Taylor expansion of $\tilde P_{m}^{(k)}$ at $\bm \omega_{m}^{(k,q)}$ must be more accurate than $\tilde P_{2,m}^{(k)}$ with
\begin{equation}
    \tilde P_{3,m}^{(k)} = \tilde P_{2,m}^{(k)} + {\frac{1}{{3!}}{{\left[ {\sum\limits_{s \in {\mathcal S}^I} {\left( {{\omega _{sm}} - \omega _{sm}^{(k,q)}} \right)\frac{\partial }{{\partial \omega _{sm}^{(k,q)}}}} } \right]}^3}{\tilde P_m^{(k)}}|_{\bm \omega_m^{(k,q)}}}.
\end{equation}

Since the error of $\tilde P_{2,m}^{(k)}$ is not greater than the maximum difference between $\tilde P_{3,m}^{(k)}$ and $\tilde P_{2,m}^{(k)}$, we have
\begin{equation}\label{eq:error_R2}
\begin{array}{l}
    R_{2}(\bm \omega_m) = \max \{ \frac{1}{{3!}}{{\left[ {\sum\limits_{s \in {\mathcal S}^I} {\left( {{\omega _{sm}} - \omega _{sm}^{(k,q)}} \right)\frac{\partial }{{\partial \omega _{sm}^{(k,q)}}}} } \right]}^3} \times \\
    \qquad \quad {\tilde P_m^{(k)}}|_{\bm \omega_m^{(k,q)}} \}.
\end{array}
\end{equation}

In (\ref{eq:error_R2}), $\bm \omega_{m}^{(k,q)}$ is a constant vector, the $\max$ operation will not affect the constant vector and the vector $\bm \omega_{m}$. For any $s\in \mathcal{S}^I$, the maximum value obtainable by $\frac{{{\partial ^3}{\tilde P_m^{(k)}}|_{\bm \omega_m^{(k,q)}}}}{{\partial \omega _{sm}^{3(k,q)}}}$ will not exceed the greatest value of that derivative in the interval $[\hat \omega_{sm}^{lb}, S_{sm}^{\star}]$. Additionally, the maximum value of $\frac{{{\partial ^3}{P_m}|_{\bm \omega_m^{(k,q)}}}}{{\partial \omega _{sm}^{3(k,q)}}}$ will generally occur at one of the endpoints of the interval $[\hat \omega_{sm}^{lb}, S_{sm}^{\star}]$. Therefore, we obtain (\ref{eq:error_bound_lemma}). This completes the proof.

\subsection{Proof of Lemma 4}
For all $i \in {\mathcal{I}^u}$, $s \in {\mathcal{S}^I}$, $m \in {\mathcal{M}}$, a feasible way of proving that ${\rm rank}(\bm G_{i,sm}) \le 1$ is to utilize the Lagrange method. However, owing to the complicated expression of $W^{u}(\bm r_m)$ w.r.t. $\bm G_{i,sm}$, it will be uneasy to do that. Fortunately, we find that the proof can be conducted if a family of auxiliary variables is introduced.

For the constraint (\ref{eq:SAA_admm_problem}d), if we introduce the auxiliary variables $\{\nu_{i,sm}\}$ and let
\begin{equation}\label{eq:linear_snr_tau_ism}
    \frac{{\rm{tr}}({\bm H_{i,sm}}{{\bm G}_{i,sm}})}{\phi \sigma _{i,s}^2} \ge \nu _{i,sm},\forall i \in {\mathcal I}_s^u,s \in {{\mathcal S}^u},m \in {\mathcal{M}},
\end{equation}
then (\ref{eq:SAA_admm_problem}d) is equivalent to
\begin{equation}\label{eq:}
    \sum\limits_{s \in {{\mathcal S}^I}} {(1+\alpha_g)\omega _{sm}(\bar t)}  + W^u(\bm f_m)  \le W, \text{ } {\rm and} \text{ } {\rm (\ref{eq:linear_snr_tau_ism})},
\end{equation}
where $\bm f_m = \{f_{i,sm};i \in {\mathcal{I}_s^u}, s \in {\mathcal{S}^u}\}$ and
\begin{equation}
\begin{array}{*{20}{l}}
{{f_{i,sm}} = \frac{{L_{i,s}^u}}{{{{\log }_2}(1 + {\nu _{i,sm}})}} + \frac{{{{({Q^{ - 1}}(\beta))}^2}}}{{2\log _2^2(1 + {\nu _{i,sm}})}}} + \\
\quad \frac{{{{({Q^{ - 1}}(\beta))}^2}}}{{2\log _2^2(1 + {\nu _{i,sm}})}}\sqrt {1 + \frac{{4L_{i,s}^u{{\log }_2}(1 + {\nu _{i,sm}})}}{{{{({Q^{ - 1}}(\beta))}^2}}}}.
\end{array}
\end{equation}

We omit the proof of the equivalence as a similar proof can be found in the proof section of constraints' equivalence in \cite{How2019yang}.

The partial Lagrangian function of (\ref{eq:SCA_bandwidth_beamforming_problem}) can be written as
\begin{equation}\label{eq:lagrangian_func}
\begin{array}{l}
L( \ldots ) = \sum\limits_{s \in {{\mathcal S}^u}} \sum\limits_{i \in {\mathcal I}_s^u} [ {\frac{{\tilde \rho \eta }}{M}{\rm{tr}}({\bm G_{i,sm}}) + \sum\limits_{j \in {\mathcal J}} {{{\bar \lambda }_{jm}}{\rm{tr}}(b_{i,sm}^{u(k,q)}{\bm Z_j}{\bm G_{i,sm}})} } \\
\quad - {{{\bar \mu }_{i,sm}}\frac{{{\rm{tr}}({\bm H_{i,sm}}{\bm G_{i,sm}})}}{{\phi \sigma _{i,s}^2}} - {\rm tr}({\bm {\bar X}_{i,sm}^{\rm H}} {\bm G_{i,sm}})} ],
\end{array}
\end{equation}
where $\bar \lambda_{jm}$, $\bar \mu_{i,sm}$, and $\bar {\bm X}_{i,sm}$ are Lagrangian multipliers corresponding to constraints (\ref{eq:SAA_admm_problem}c), (\ref{eq:linear_snr_tau_ism}) and (\ref{eq:SAA_admm_problem}e). Besides, only terms related to ${\bm G}_{i,sm}$ are included in this function for brevity.

According to the Karush-Kuhn-Tucker (KKT) conditions, the necessary condition for obtaining the optimal matrix power at the $(k,q)$-th iteration ${\bm G}_{i,sm}^{(k,q)\star}$ is given by
\begin{equation}\label{eq:KKT_condition_2}
    \begin{array}{l}
\frac{{\partial L( \ldots )}}{{\partial {\bm G}_{i,sm}^{(k,q)\star}}} =  \frac{{\tilde \rho \eta }}{M}{\bm I_{i,sm}} + \sum\limits_{j \in {\mathcal J}} {{{\bar \lambda }_{jm}}{b_{i,sm}^{u(k,q)}\bm Z_j}} - \\
\qquad \frac{\bar \mu_{i,sm}{{\bm H_{i,sm}}}}{{\phi \sigma _{i,s}^2}} - {{\bm X}_{i,sm}} = 0,
\end{array}
\end{equation}
where ${\bm I_{i,sm}} \in \mathbb{R}^{JK \times JK}$ is an identity matrix.

Then, we can conclude that ${\rm rank}({\bm X}_{i,sm}) \ge JK - 1$. The reasons are i) ${{\bar \lambda }_{jm}}$, $b_{i,sm}^{u(k,q)}$, and $\bar \mu _{i,sm}$ are nonnegative and the matrix $\bm I_{i,sm}$ is full rank; ii) ${\rm rank}(\bm H_{i,sm}) \le 1$.

Next, according to the complementary slackness condition, we have
\begin{equation}\label{eq:complementary_slack_condition_2}
    {{\bm X}_{i,sm}}{\bm G}_{i,sm}^{(k,q)\star} = 0.
\end{equation}

Based on (\ref{eq:complementary_slack_condition_2}) and the rank result of $\bm X_{i,sm}$, we can conclude that ${\rm rank}(\bm G_{i,sm}^{(k,q)\star}) \le 1$. This completes the proof.

 use section* for acknowledgment
\section*{Acknowledgment}
The authors are grateful for Prof. Chi Harold Liu who helps them a lot on the construction of neural networks and the mitigation of many other questions.

\ifCLASSOPTIONcaptionsoff
  \newpage
\fi




%
\bibliographystyle{IEEEtran}
\bibliography{Network_slicing}

\end{document}